\documentclass[11pt]{article}
\usepackage{CJK}
\usepackage{latexsym,bm}
\usepackage{xypic}
\usepackage{amsmath}
\usepackage{wasysym}
\usepackage{indentfirst}
\usepackage{amssymb}
\usepackage{dsfont}
\usepackage{amsthm}
\begin{document}
\newcommand{\fr}[2]{\frac{\;#1\;}{\;#2\;}}
\newtheorem{theorem}{Theorem}[section]
\newtheorem{lemma}{Lemma}[section]
\newtheorem{proposition}{Proposition}[section]
\newtheorem{corollary}{Corollary}[section]
\newtheorem{remark}{Remark}[section]
\newtheorem{definition}{Definition}[section]
\newtheorem{notation}{Notation}[section]
\newtheorem{example}{Example}[section]
\numberwithin{equation}{section}
\newcommand{\Aut}{\mathrm{Aut}\,}
\newcommand{\fun}{\mathrm{fun}\,}
\newcommand{\supp}{\mathrm{supp}\,}
\newcommand{\rk}{\mathrm{rk}\,}
\newcommand{\rank}{\mathrm{rank}\,}
\newcommand{\Mat}{\mathrm{Mat}\,}
\newcommand{\col}{\mathrm{col}\,}
\newcommand{\length}{\mathrm{length}\,}
\newcommand{\leftlen}{\mathrm{leftlen}\,}
\newcommand{\rightlen}{\mathrm{rightlen}\,}
\newcommand{\rightann}{\mathrm{rightann}\,}
\newcommand{\leftann}{\mathrm{leftann}\,}
\newcommand{\dom}{\mathrm{dom}\,}
\newcommand{\len}{\mathrm{len}\,}
\newcommand{\CSupp}{\mathrm{CSupp}\,}
\newcommand{\wt}{\mathrm{wt}\,}
\newcommand{\Rank}{\mathrm{Rank}\,}
\newcommand{\tr}{\mathrm{tr}\,}
\newcommand{\Rk}{\mathrm{Rk}\,}
\newcommand{\lcm}{\mathrm{lcm}\,}

\newcommand{\Hom}{\mathrm{Hom}\,}
\newcommand{\End}{\mathrm{End}\,}
\title{A Galois Connection Approach to Wei-Type Duality Theorems$^\ast$}
\author{Yang Xu$^1$ \,\,\,\,\,\, Haibin Kan$^2$\,\,\,\,\,\,Guangyue Han$^3$}
\maketitle

\renewcommand{\thefootnote}{\fnsymbol{footnote}}

% 单位、地址、基金
\footnotetext{\hspace*{-12mm} \begin{tabular}{@{}r@{}p{13.4cm}@{}}
$^\ast$ & This work has been submitted to the IEEE for possible publication. Copyright may be transferred without notice, after which this version may no longer be accessible.\\
$^1$ & Shanghai Key Laboratory of Intelligent Information Processing, School of Computer Science, Fudan University,
Shanghai 200433, China.\\
&Department of Mathematics, Faculty of Science, The University of Hong Kong, Pokfulam Road, Hong Kong, China. {E-mail:12110180008@fudan.edu.cn} \\
$^2$ & Shanghai Key Laboratory of Intelligent Information Processing, School of Computer Science, Fudan University,
Shanghai 200433, China. {E-mail:hbkan@fudan.edu.cn} \\
$^3$ & Department of Mathematics, Faculty of Science, The University of Hong Kong, Pokfulam Road, Hong Kong, China. {E-mail:ghan@hku.hk} \\

\end{tabular}}

\vskip 3mm

 {\hspace*{-6mm}\bf Abstract---}\! In $1991$, Wei proved a duality theorem that established an interesting connection between the generalized Hamming weights of a linear code and those of its dual code. Wei's duality theorem has since been extensively studied from different perspectives and extended to other settings. In this paper, we re-examine Wei's duality theorem and its various extensions, henceforth referred to as Wei-type duality theorems, from a new Galois connection perspective. Our approach is based on the observation that the generalized Hamming weights and the dimension/length profiles of a linear code form a Galois connection. The central result of this paper is a general Wei-type duality theorem for two Galois connections between finite subsets of $\mathbb{Z}$, from which all the known Wei-type duality theorems can be recovered. As corollaries of our central result, we prove new Wei-type duality theorems for $w$-demi-matroids defined over finite sets and $w$-demi-polymatroids defined over modules with a composition series, which further allows us to unify and generalize all the known Wei-type duality theorems established for codes endowed with various metrics.

\section{Introduction}

Throughout the paper, we let $\mathds{Z}$ and $\mathds{Z}^{+}$ denote the set of all the integers and all the positive integers, respectively. Furthermore, we let $\mathds{N}=\mathds{Z}^{+}\cup\{0\}$. For any $a,b\in\mathds{Z}$, we use $[a,b]$ to denote the set of all the integers between $a$ and $b$, i.e.,
$$[a,b]=\{i\mid i\in \mathds{Z},~a\leqslant i\leqslant b\}.$$
Note that if $a\geqslant b+1$, then $[a,b]=\emptyset$.

Motivated by applications from cryptography, in 1991, Wei proposed and studied generalized Hamming weights (GHWs) of linear codes. To be more specific, we let $\mathds{F}$ be a field and let $C$ be an $(m,k)$ linear code over $\mathds{F}$, i.e., $C$ is an $\mathds{F}$-subspace of $\mathds{F}^{m}$ with $\dim_{\mathds{F}}(C)=k$. Following [40, Section II], for any $r\in[0,k]$, \textit{the $r$-th generalized Hamming weight of $C$}, denoted by $\mathbf{d}_{r}(C)$, is defined as
\begin{equation}\mathbf{d}_{r}(C)=\min\{|\chi(D)|\mid \text{$D$ is an $(m,r)$ subcode of $C$}\},\end{equation}
where for any $D\subseteq \mathds{F}^{m}$,
\begin{equation}\chi(D)=\{i\mid i\in[1,m],~\alpha_{(i)}\neq0~\text{for some $\alpha\in D$}\}\end{equation}
denotes the set of not-always-zero bit positions of $D$. Note that the definition of GHWs applies to the $(m, m-k)$ linear code $C^{\bot}$ (the dual code of $C$) as well. It has been shown in [40, Appendix] that the GHWs of $C^{\bot}$ characterize both the performance of $C$ on the wire-tap channel of type II and the performance of $C$ as a $t$-resilient function (see \cite{10,34}). GHWs have been widely used to list decoding from erasures and to gauge the security performances of linear codes for secret sharing, secure network coding or distributed data storage. Relevant work for linear codes can be found in \cite{5,16,19,28}, some of which have been further extended to rank metric codes \cite{27,31}.

Wei has proved a duality theorem that establishes an interesting connection between the GHWs of $C$ and those of its dual code $C^{\bot}$. To be more specific, it has been proven in [40, Theorem 3] that the following two sets
\begin{equation} \label{Wei-Result}
\{\mathbf{d}_{r}(C)\mid r\in[1,k]\},~\{m+1-\mathbf{d}_{r}(C^{\bot})\mid r\in[1,m-k]\},
\end{equation}
form a partition of $[1,m]$, and consequently, the GHWs of $C$ and those of $C^{\bot}$ determine each other.

A closely related study can be found in Forney's $1994$ paper~\cite{16}. More specifically, following [16, Section 3], for any $l\in[0,m]$, the $l$-th dimension/length profile (DLP) of $C$, denoted by $\mathbf{K}_{l}(C)$, is defined as
\begin{equation}
\mathbf{K}_{l}(C)=\max\{\dim_{\mathbb{F}}(C\cap\delta(J))\mid J\subseteq[1,m],~|J|=l\},
\end{equation}
where for any $J\subseteq [1,m]$,
\begin{equation}\delta(J) =\{\alpha\mid \alpha \in \mathbb{F}^m~s.t.~\forall~i\in[1,m]-J,~\alpha_{(i)}=0\}\end{equation}
denotes the set of all the elements of $\mathbb{F}^m$ whose positions outside of $J$ are all zeros. Forney has shown that
\begin{equation} \label{Forney-Result}
\text{$\mathbf{K}_{l}(C^{\bot})=\mathbf{K}_{m-l}(C)+l-k$ for all $l\in[0,m]$},
\end{equation}
from which Wei's duality theorem can be recovered (see Theorems $2$, $3$ and $4$ of [16] for more details).

Wei's duality theorem has since been generalized and extended in a number of directions. For codes over finite rings, Ashikhmin proves \cite{2} a Wei-type duality theorem for linear codes over Galois rings. Horimoto and Shiromoto prove \cite{23} a Wei-type duality theorem for linear codes over finite chain rings. For codes with a poset metric (see \cite{9}, where ``poset'' is short for the term ``partially ordered set''), Wei-type duality theorems have been proven by Barg and Purkayastha in \cite{3} and by Moura and Firer in \cite{32}. For rank metric codes (see \cite{13,17,20,37}), Wei-type duality theorems have been proven by Ravagnani in \cite{36}, by Ducoat in \cite{15} and by Mart\'{\i}nez-Pe\~{n}as and Matsumoto in \cite{31}, each in terms of a similar yet different version of generalized rank weights. A Wei-type duality theorem for sum-rank metric codes has been established by Mart\'{\i}nez-Pe\~{n}as in \cite{30}.

Wei-type duality theorems have also been proven for some combinatorial notions. In \cite{6}, Britz, Johnsen, Mayhew and Shiromoto prove two Wei-type duality theorems for demi-matroids, and as consequences of these results, they further derive Wei-type theorems for matroids, graphs and transversals. In \cite{18}, Ghorpade and Johnsen prove Wei-type duality theorems for demi-polymatroids, a generalization of the $q$-analogue of a matroid defined over vector spaces (see \cite{21,26,38}), and as a consequence, they obtain a Wei-type duality theorem for rank metric code flags. A similar approach to \cite{18} for demi-polymatroids has been independently proposed by Britz, Mammomiti and Shiromoto in \cite{7}. Recently, Panja, Pratihar and Hajatiana Randrianarisoa have proved \cite{35} a Wei-type duality theorem for sum-matroids, which is another generalization of the $q$-analogue of a matroid.

In this paper, we will re-examine Wei's original duality theorem and more generally Wei-type duality theorems from a Galois connection perspective. Our starting point is the observation that the GHWs and DLPs of a linear code form a Galois connection, for which Wei's duality theorem holds; and moreover, Galois connections also arise naturally in similar or more general settings, where various Wei-type duality theorems hold.

In Section 2, we show that the GHWs and DLPs of a linear code form a Galois connection, and then we prove our main result Theorem 2.2, which is a Wei-type duality theorem for two Galois connections between finite subsets of $\mathds{Z}$. This result is of central importance in the sense that it implies all the known Wei-type duality theorems, and moreover, it can be applied to derive new Wei-type duality theorems, as detailed in later sections. In Section 3, as corollaries of Theorem 2.2, we establish Theorems 3.1 and 3.2, two bridging theorems that can facilitate the application of Theorem 2.2. Indeed, Theorems 3.1 and 3.2 can be used to recover known or derive new Wei-type duality theorems via simple substitution.

In Section 4, we prove Wei-type duality theorems for generalized weights and profiles of $w$-demi-matroids defined over finite sets (Theorems 4.1 and 4.2). These results generalize the corresponding results for demi-matroids in \cite{6}. In Section 5, we prove a Wei-type duality theorem for generalized weights and profiles of $w$-demi-polymatroids defined over modules with a composition series (Theorem 5.1). Our result generalizes the corresponding results for $(q,m)$-demi-polymatroids defined over vector spaces in [7, 18].

Sections 6 and 7 are devoted to generalized weights, profiles and Wei-type duality theorems for codes endowed with various metrics, more specifically, Gabidulin-Roth rank metric (\cite{15,17,27,37}), poset metric (\cite{9,32}), Delsarte rank metric (\cite{13,20,22,31,36}) and generalized Hamming weights with respect to rank of modules (\cite{23}). Following the framework of \cite{39}, we consider codes over modules with a composition series. We show that generalized weights and profiles of such codes can be rephrased in terms of the associated $w$-demi-matroids or $w$-demi-polymatroids; and furthermore, dual $w$-demi-matroids or $w$-demi-polymatroids are always associated to dual codes. Based on these observations, we prove Wei-type duality theorems for codes endowed with different metrics through a unified approach in Theorems 7.1--7.4, which respectively generalizes the corresponding existing Wei-type duality theorems.

\section{Main Result}

\setlength{\parindent}{2em}

\subsection{Basics on Galois connections}

\setlength{\parindent}{2em}
In this subsection, we collect some basic facts on Galois connections between finite subsets of $\mathds{Z}$ with respect to the order $\leqslant$. Here we note that the notion of Galois connection can be defined more generally for posets; see, e.g., [12, Definition 7.23] and [4, Page 124].

Throughout this subsection, we let $P$ and $Q$ be nonempty finite subsets of $\mathds{Z}$, and we begin by recalling the definition of Galois connection between $P$ and $Q$.

\setlength{\parindent}{0em}
\begin{definition}
{Given $\varphi:P\longrightarrow Q$ and $\psi:Q\longrightarrow P$, $(\varphi,\psi)$ is said to be a Galois connection between $P$ and $Q$ if the following two conditions hold:

{\bf{(1)}}\,\,Both $\varphi$ and $\psi$ preserve the order $\leqslant$;

{\bf{(2)}}\,\,For any $(a,b)\in P\times Q$, we have $a\leqslant\psi(b)\Longleftrightarrow\varphi(a)\leqslant b$.
}
\end{definition}

\setlength{\parindent}{2em}
The following basic facts on Galois connections will be used frequently.

\setlength{\parindent}{0em}
\begin{lemma}
{Let $(\varphi,\psi)$ be a Galois connection between $P$ and $Q$. Then, it holds that:

{\bf{(1)}}\,\,For any $\lambda\in P$, $\varphi(\lambda)=\min\{b\mid b\in Q,~\lambda\leqslant\psi(b)\}$;

{\bf{(2)}}\,\,For any $\mu\in Q$, $\psi(\mu)=\max\{a\mid a\in P,~\varphi(a)\leqslant\mu\}$;

{\bf{(3)}}\,\,Let $d_0=\min(Q)$. Then, $\varphi^{-1}[\{d_0\}]=\{a\mid a\in P,~a\leqslant\psi(d_0)\}$;

{\bf{(4)}}\,\,Let $d\in Q$ where $d\neq\min(Q)$, and let $v=\max\{b\mid b\in Q,~b\leqslant d-1\}$. \hspace*{5mm}\,\,Then, for any $a\in P$, we have $d=\varphi(a)\Longleftrightarrow\psi(v)+1\leqslant a\leqslant \psi(d)$.
}
\end{lemma}

\begin{proof}
(1) and (2) are special cases of [12, Proposition 7.31]. Moreover, (3) and (4) follow from (1) and (2) via a routine verification.
\end{proof}

\begin{remark}
We will show in Section 2.2 that the GHWs and DLPs of a linear code form a Galois connection. As a result, some relevant results on linear codes (e.g., [16, Section III], [19, Lemma 1], [40, Corollary A]) can be reformulated in terms of Galois connections and thereby can be proven by using Lemma 2.1. Moreover, as detailed in later sections, Galois connections naturally arise in other settings, and some relevant results on rank metric codes (e.g., [31, Proposition 14], [36, Theorem 42]) or on $(q,m)$-demi-polymatroids (e.g., [18, Lemma 12]) also follow from Lemma 2.1 after appropriate reformulation.
\end{remark}

\setlength{\parindent}{2em}
We end this subsection with the following lemma, whose proof is straightforward and thus omitted.

\setlength{\parindent}{0em}
\begin{lemma}
{For a (possibly infinite) set $X$, let $f:X\longrightarrow P$, $g:X\longrightarrow Q$ such that $\max(P)\in f[X]$, $\min(Q)\in g[X]$. Define $\varphi:P\longrightarrow Q$ as
$$\varphi(a)=\min\{g(u)\mid u\in X,~a\leqslant f(u)\},$$
and define $\psi:Q\longrightarrow P$ as
$$\psi(b)=\max\{f(u)\mid u\in X,~g(u)\leqslant b\}.$$
Then, $(\varphi,\psi)$ is a Galois connection between $P$ and $Q$.
}
\end{lemma}

\subsection{Galois connections arising from GHWs and DLPs}

\setlength{\parindent}{2em}
In this subsection, we show that the GHWs and DLPs of a linear code form a Galois connection, which is a key observation underpinning our treatment of Wei-type duality theorems. Recall that in Section 1, for an $(m,k)$ linear code $C$ over a field $\mathbb{F}$, $\mathbf{d}_{r}(C)$ denotes its $r$-th GHW and $\mathbf{K}_{l}(C)$ its $l$-th DLP.

\setlength{\parindent}{0em}
\begin{theorem} \label{observation}
Define $\varphi:[0,k]\longrightarrow[0,m]$ as $\varphi(r)=\mathbf{d}_{r}(C)$, and define $\psi:[0,m]\longrightarrow[0,k]$ as $\psi(l)=\mathbf{K}_{l}(C)$. Then, $(\varphi,\psi)$ is a Galois connection between $[0,k]$ and $[0,m]$.
\end{theorem}

\begin{proof}
By [40, Theorem 2], for any $r\in[0,k]$, we have
$$\varphi(r)=\mathbf{d}_{r}(C)=\min\{|J|\mid J\subseteq[1,m],~r\leqslant\dim_{\mathds{F}}(C\cap\delta(J))\},$$
where for any $J\subseteq [1,m]$, $\delta(J)$ is defined as in (1.5). Moreover, it is straightforward to verify that for any $l\in[0,m]$, it holds that
$$\psi(l)=\mathbf{K}_{l}(C)=\max\{\dim_{\mathds{F}}(C\cap\delta(J))\mid J\subseteq[1,m],~|J|\leqslant l\}.$$
Now applying Lemma 2.2 with $f:2^{[1,m]}\longrightarrow[0,k]$ set to be $f(J)=\dim_{\mathds{F}}(C\cap\delta(J))$, and $g:2^{[1,m]}\longrightarrow[0,m]$ set to be $g(J)=|J|$, we conclude that $(\varphi,\psi)$ is a Galois connection between $[0,k]$ and $[0,m]$.
\end{proof}

\subsection{The central theorem}

\setlength{\parindent}{2em}
Throughout this subsection, we let $(k,m)\in\mathds{N}^{2}$, $w\in\mathds{Z}^{+}$.

\setlength{\parindent}{0em}
\begin{lemma}
{Let $(\varphi,\psi)$ be a Galois connection between $[0,k]$ and $[0,m]$ such that $\psi(0)=0$ and $\psi(l)-\psi(l-1)\leqslant w$ for any $l\in[1,m]$. Define $\eta:[0,m]\longrightarrow \mathds{Z}$ as $\eta(l)=\psi(m-l)+wl-k$. Then, we have:

{\bf{(1)}}\,\,$\eta(0)=0$, $\eta(m)=wm-k$;

{\bf{(2)}}\,\,For any $l\in[1,m]$, $0\leqslant\eta(l)-\eta(l-1)\leqslant w$;

{\bf{(3)}}\,\,There exists $\tau:[0,wm-k]\longrightarrow[0,m]$ such that $(\tau,\eta)$ is a Galois connection between $[0,wm-k]$ and $[0,m]$. Moreover, for any $u\in[0,k]$, $v\in[0,wm-k]$ with $\varphi(u)+\tau(v)=m+1$, it holds true that $u\not\equiv v+k~(\bmod~w)$.
}
\end{lemma}

\begin{proof}
(1) and (2) follow from straightforward computation, and so we only prove (3). Define $\tau:[0,wm-k]\longrightarrow[0,m]$ as
$$\tau(a)=\min\{b\mid b\in[0,m],a\leqslant\eta(b)\}.$$
Then, it can be readily verified that $\tau$ is well defined, and $(\tau,\eta)$ is a Galois connection between $[0,wm-k]$ and $[0,m]$. Now for $u\in[0,k]$, $v\in[0,wm-k]$ with $\varphi(u)+\tau(v)=m+1$, note that $\tau(v)=m+1-\varphi(u)$, $\varphi(u)\in[1,m]$, we apply (4) of Lemma 2.1 to $(\tau,\eta)$ and $(v,m+1-\varphi(u))$ and reach
$$\eta(m-\varphi(u))+1\leqslant v \leqslant \eta(m+1-\varphi(u)),$$
which, together with the definition of $\eta$, further implies that
$$\psi(\varphi(u))+w(m-\varphi(u))-k+1\leqslant v\leqslant\psi(\varphi(u)-1)+w(m+1-\varphi(u))-k.$$
Again by Lemma 2.1, we have $u\leqslant\psi(\varphi(u))$ and $\psi(\varphi(u)-1)\leqslant u-1$. It then follows that
$$u+w(m-\varphi(u))-k+1\leqslant v\leqslant u-1+w(m+1-\varphi(u))-k,$$
which further implies that
\begin{equation}1\leqslant v+k-u-w(m-\varphi(u))\leqslant w-1.\end{equation}
By (2.1), we conclude that $u\not\equiv v+k~(\bmod~w)$, completing the proof.
\end{proof}

\setlength{\parindent}{2em}
\begin{remark}
Lemma 2.3 is largely inspired by [36, Theorem 37], and similar results for rank metric codes and for $(q,m)$-demi-polymatroids have also been established in [31, Lemma 66] and [18, Theorem 15], respectively.
\end{remark}

\setlength{\parindent}{2em}
The following proposition ought to be known, however we are not able to locate a reference and hence a proof is included for completeness.

\setlength{\parindent}{0em}
\begin{proposition}
{Suppose that $(\varphi,\psi)$ is a Galois connection between $[0,k]$ and $[0,m]$. Then, the following three statements are equivalent to each other:

{\bf{(1)}}\,\,$\psi(l)-\psi(l-1)\leqslant w$ for all $l\in[1,m]$;

{\bf{(2)}}\,\,$|\varphi^{-1}[\{l\}]|\leqslant w$ for all $l\in[1,m]$;

{\bf{(3)}}\,\,$\varphi(r)+1\leqslant \max\{\varphi(r+w),1\}$ for all $r\in[0,k-w]$.

Moreover, if $\psi(0)=0$, then $\varphi(a)\in[1,m]$ for all $a\in[1,k]$, and (1)--(3) are equivalent to the following

{\bf{(4)}}\,\,$\varphi(r)+1\leqslant \varphi(r+w)$ for all $r\in[0,k-w]$.
}
\end{proposition}

\begin{proof}
By (4) of Lemma 2.1, we have $|\varphi^{-1}[\{l\}]|=\psi(l)-\psi(l-1)$ for all $l\in[1,m]$, which immediately implies $(1)\Longleftrightarrow(2)$. Next, we prove $(1)\Longleftrightarrow(3)$.

$(1)\Longrightarrow(3)$ Consider $r\in[0,k-w]$. If $\varphi(r+w)\leqslant0$, then $\varphi(r)+1\leqslant\varphi(r+w)+1\leqslant1$, which implies (3), as desired. Therefore in the following, we assume that $1\leqslant\varphi(r+w)$. By (1), we have $\psi(\varphi(r+w))-\psi(\varphi(r+w)-1)\leqslant w$. By (1) of Lemma 2.1, we deduce that $r+w\leqslant\psi(\varphi(r+w))$, which further implies that $r\leqslant\psi(\varphi(r+w)-1)$. Now, applying Definition 2.1 to $(r,\varphi(r+w)-1)$, we have $\varphi(r)\leqslant\varphi(r+w)-1$, which again implies (3).

$(3)\Longrightarrow(1)$ Consider $l\in[1,m]$. If $\psi(l)\leqslant w$, then by $0\leqslant\psi(l-1)$, we deduce that $\psi(l)-\psi(l-1)\leqslant w$, which implies (1), as desired. Therefore in the following, we assume that $\psi(l)\geqslant w$. It then follows that $\psi(l)-w\in[0,k-w]$, which, together with (3), implies that $\varphi(\psi(l)-w)+1\leqslant\max\{\varphi(\psi(l)),1\}$. By (2) of Lemma 2.1, we have $\varphi(\psi(l))\leqslant l$, which, together with $l\geqslant1$, implies that $\varphi(\psi(l)-w)+1\leqslant l$, and hence $\varphi(\psi(l)-w)\leqslant l-1$. Applying Definition 2.1 to $(\psi(l)-w,l-1)$, we conclude that $\psi(l)-w\leqslant\psi(l-1)$, which again implies (1).

\hspace*{6mm}The remainder of the proposition follows from (3) of Lemma 2.1 and the proven fact that $(1)\Longleftrightarrow(3)$.
\end{proof}

\setlength{\parindent}{2em}
The proof of the following lemma is straightforward and thus omitted.

\setlength{\parindent}{0em}
\begin{lemma}
{Assume that $wm\geqslant k$. For any $\gamma\in\mathbb{Z}$, define the sets $\mathcal{U}_{(\gamma)}$ and $\mathcal{V}_{(\gamma)}$ as
$$\mathcal{U}_{(\gamma)}=\{u\mid u\in[1,k],u\equiv \gamma+k~(\bmod ~w)\},$$
$$\mathcal{V}_{(\gamma)}=\{v\mid v\in[1,wm-k],v\equiv \gamma~(\bmod ~w)\}.$$
Then, for any $\gamma\in\mathbb{Z}$, we have $|\mathcal{U}_{(\gamma)}|+|\mathcal{V}_{(\gamma)}|=m$.
}
\end{lemma}

\setlength{\parindent}{2em}
We are now ready to state and prove the main result of this paper.

\setlength{\parindent}{0em}
\begin{theorem} \label{central}
{Let $(\varphi,\psi)$ be a Galois connection between $[0,k]$ and $[0,m]$ such that $\psi(0)=0$ and $\psi(l)-\psi(l-1)\leqslant w$ for all $l\in[1,m]$. Let $(\tau,\eta)$ be a Galois connection between $[0,wm-k]$ and $[0,m]$. For any $\gamma\in\mathds{Z}$, define the sets $\mathcal{A}_{(\gamma)}$ and $\mathcal{B}_{(\gamma)}$ as
$$\mathcal{A}_{(\gamma)}=\{\varphi(u)\mid u\in[1,k],~u\equiv\gamma+k~(\bmod~w)\},$$
$$\mathcal{B}_{(\gamma)}=\{m+1-\tau(v)\mid v\in[1,wm-k],~v\equiv\gamma~(\bmod~w)\}.$$
Then, the following four statements are equivalent to each other:

{\bf{(1)}}\,\,$\eta(l)=\psi(m-l)+wl-k$ for all $l\in[0,m]$;

{\bf{(2)}}\,\,$\eta(0)=0$, $\eta(l)-\eta(l-1)\leqslant w$ for all $l\in[1,m]$, and for any $(u,v)\in[1,k]\times[1,wm-k]$, $\varphi(u)+\tau(v)=m+1\Longrightarrow u\not\equiv~v+k~(\bmod~w)$;

{\bf{(3)}}\,\,For any $\gamma\in\mathds{Z}$, $\mathcal{A}_{(\gamma)}\cap \mathcal{B}_{(\gamma)}=\emptyset$, $\mathcal{A}_{(\gamma)}\cup \mathcal{B}_{(\gamma)}=[1,m]$;

{\bf{(4)}}\,\,For any $\gamma\in\mathds{Z}$, $\mathcal{A}_{(\gamma)}\cup \mathcal{B}_{(\gamma)}=[1,m]$.
}
\end{theorem}

\begin{proof}
As detailed below, the proof consists of the following $5$ steps. First of all, for any $\gamma\in\mathbb{Z}$, we define $\mathcal{U}_{(\gamma)}=\{u\mid u\in[1,k],u\equiv \gamma+k~(\bmod ~w)\}$, $\mathcal{V}_{(\gamma)}=\{v\mid v\in[1,wm-k],v\equiv \gamma~(\bmod ~w)\}$. By Lemma 2.4, for any $\gamma\in\mathbb{Z}$, it holds that $|\mathcal{U}_{(\gamma)}|+|\mathcal{V}_{(\gamma)}|=m$.

$(1)\Longrightarrow(2)$\,\,This follows from Lemma 2.3.

$(2)\Longrightarrow(3)$\,\,Fix  $\gamma\in\mathds{Z}$. For any $u\in \mathcal{U}_{(\gamma)}$ and $v\in \mathcal{V}_{(\gamma)}$, noticing that $u\equiv~v+k~(\bmod~w)$, by (2), we deduce that $\varphi(u)\neq m+1-\tau(v)$. It then follows that $\mathcal{A}_{(\gamma)} \cap \mathcal{B}_{(\gamma)}=\emptyset$. By Proposition 2.1, we deduce that $\varphi(a)\in[1,m]$ for all $a\in[1,k]$, and $\varphi(r)+1\leqslant \varphi(r+w)$ for all $r\in[0,k-w]$, which further imply that $\mathcal{A}_{(\gamma)}\subseteq[1,m]$ and $|\mathcal{A}_{(\gamma)}|=|\mathcal{U}_{(\gamma)}|$. Applying Proposition 2.1 to $(\tau,\eta)$ with $k$ replaced by $wm-k$, we deduce that $\tau(a)\in[1,m]$ for all $a\in[1,wm-k]$, and $\tau(r)+1\leqslant \tau(r+w)$ for all $r\in[0,w(m-1)-k]$, which further imply that $\mathcal{B}_{(\gamma)}\subseteq[1,m]$ and $|\mathcal{B}_{(\gamma)}|=|\mathcal{V}_{(\gamma)}|$. It then follows that $|\mathcal{A}_{(\gamma)}|+|\mathcal{B}_{(\gamma)}|=|\mathcal{U}_{(\gamma)}|+|\mathcal{V}_{(\gamma)}|=m$, which, together with $\mathcal{A}_{(\gamma)} \cap \mathcal{B}_{(\gamma)}=\emptyset$, implies that $|\mathcal{A}_{(\gamma)} \cup \mathcal{B}_{(\gamma)}|=m$. Finally, using the fact that $\mathcal{A}_{(\gamma)} \cup \mathcal{B}_{(\gamma)} \subseteq [1,m]$, we deduce that $\mathcal{A}_{(\gamma)} \cup \mathcal{B}_{(\gamma)}=[1,m]$, as desired.

$(3)\Longrightarrow(4)$\,\,This is trivial.

$(4)\Longrightarrow(3)$\,\,Fix $\gamma\in\mathds{Z}$. From the facts that $|\mathcal{A}_{(\gamma)}|\leqslant|\mathcal{U}_{(\gamma)}|$, $|\mathcal{B}_{(\gamma)}|\leqslant|\mathcal{V}_{(\gamma)}|$, $|\mathcal{U}_{(\gamma)}|+|\mathcal{V}_{(\gamma)}|=m$ and $\mathcal{A}_{(\gamma)}\cup \mathcal{B}_{(\gamma)}=[1,m]$, we deduce that $|\mathcal{A}_{(\gamma)}|+|\mathcal{B}_{(\gamma)}|\leqslant |\mathcal{A}_{(\gamma)}\cup \mathcal{B}_{(\gamma)}|$, which implies that $\mathcal{A}_{(\gamma)}\cap \mathcal{B}_{(\gamma)}=\emptyset$, and (3) immediately follows.

$(3)\Longrightarrow(1)$\,\,By Lemma 2.3, there exists a Galois connection $(\xi,\zeta)$ between $[0,wm-k]$ and $[0,m]$ such that $\zeta(l)=\psi(m-l)+wl-k$ for all $l\in[0,m]$. Applying Proposition 2.1 to $(\xi,\zeta)$ with $k$ replaced by $wm-k$, together with Lemma 2.3, we deduce that
\begin{equation}\text{$\xi(r)+1\leqslant \xi(r+w)$ for all $r\in[0,w(m-1)-k]$}.\end{equation}
Now for a fixed yet arbitrary $\gamma\in\mathds{Z}$, define $\mathcal{L}_{(\gamma)}=\{m+1-\xi(v)\mid v\in \mathcal{V}_{(\gamma)}\}$. Then, using a parallel argument as in the step of $(2) \Longrightarrow (3)$ (with $(\tau,\eta)$ replaced by $(\xi,\zeta)$), we deduce that $\mathcal{A}_{(\gamma)}\cap \mathcal{L}_{(\gamma)}=\emptyset$, $\mathcal{A}_{(\gamma)}\cup \mathcal{L}_{(\gamma)}=[1,m]$, which, together with $\mathcal{A}_{(\gamma)}\cap \mathcal{B}_{(\gamma)}=\emptyset$ and $\mathcal{A}_{(\gamma)}\cup \mathcal{B}_{(\gamma)}=[1,m]$, imply that $\mathcal{L}_{(\gamma)}=\mathcal{B}_{(\gamma)}$. It immediately follows that
$$
f[\mathcal{V}_{(\gamma)}]=\mathcal{B}_{(\gamma)}=\mathcal{L}_{(\gamma)}=g[\mathcal{V}_{(\gamma)}],
$$
where $f:\mathcal{V}_{(\gamma)}\longrightarrow\mathbb{Z}$ is defined as $f(v)=m+1-\tau(v)$ and $g:\mathcal{V}_{(\gamma)}\longrightarrow\mathbb{Z}$ is defined as $g(v)=m+1-\xi(v)$. Since both $\tau$ and $\xi$ preserves the order $\leqslant$, for any $c,d\in \mathcal{V}_{(\gamma)}$ with $c\leqslant d$, it holds true that $f(c)\geqslant f(d)$, $g(c)\geqslant g(d)$. Noting that by (2.2), $g$ is injective, we infer that $f=g$, which implies that $\tau(v)=\xi(v)$ for any $v\in[1,wm-k]$ with $v\equiv \gamma~(\bmod ~w)$. It then follows from the arbitrariness of $\gamma$ that $\tau(v)=\xi(v)$ for all $v\in[1,wm-k]$. By (3) of Lemma 2.1, we have $\tau(0)=\xi(0)=0$, which further implies that $\tau=\xi$. Since $(\tau,\eta)$ and $(\xi,\zeta)$ are Galois connections between $[0,wm-k]$ and $[0,m]$, we apply (2) of Lemma 2.1 to reach $\eta=\zeta$, and (1) immediately follows.
\end{proof}

\begin{remark}
With the help of Theorem \ref{observation}, Theorem \ref{central} can be used to show the equivalence between Wei's original duality theorem in (1.3) and Forney's result in (1.6). More specifically, set $w=1$; and as in Theorem \ref{observation}, set $\varphi, \psi$ to be the GHWs and DLPs of $C$, respectively; and moreover, set $\tau, \eta$ to be those of $C^{\bot}$, respectively. Then, by Theorem \ref{observation}, $(\varphi,\psi)$ is a Galois connection between $[0,k]$ and $[0,m]$, and $(\tau,\eta)$ is a Galois connection between $[0,m-k]$ and $[0,m]$. It can then be verified that (1) of Theorem \ref{central} boils down to (1.6), and (3) of Theorem 2.2 boils down to (1.3), thereby recovering the equivalence between Wei's duality theorem and Forney's result.

\setlength{\parindent}{2em}
Wei's duality theorem has been extensively studied and extended to other codes or even settings beyond coding theory. To the best of our knowledge, all Wei-type duality theorems, previously known in the literature or newly established in this paper, are special cases of Theorem \ref{central} with the two Galois connections appropriately set. In this sense, Theorem \ref{central} reveals the essence of Wei-type duality theorems from a Galois connection perspective.
\end{remark}

\section{Bridging theorems}

\setlength{\parindent}{2em}
In this section, to facilitate the application of Theorem 2.2, we prove two bridging theorems that can be used to recover known or derive new Wei-type duality theorems via simple substitution.

\subsection{The first bridging theorem}
\setlength{\parindent}{2em}
Throughout this subsection, we fix the following notations:

\setlength{\parindent}{0em}
\begin{itemize}
  \item $Y$ is a nonempty set, $m\in\mathds{N}$, and $g:Y\longrightarrow [0,m]$ is a surjective map.
  \item $w\in\mathds{Z}^{+}$, $k\in\mathds{N}$, and $f:Y\longrightarrow [0,k]$ is a map satisfying the following four conditions:
\begin{equation}\forall~y\in Y,~g(y)=0\Longrightarrow f(y)=0;\end{equation}
\begin{equation}\forall~y\in Y,~wg(y)-f(y)\leqslant wm-k;\end{equation}
\begin{equation}\hspace*{-6mm}\forall~u\in Y~s.t.~g(u)\leqslant m-1,~\exists~v\in Y~s.t.~g(v)=g(u)+1,~f(u)\leqslant f(v);\end{equation}
\begin{equation}\hspace*{-6mm}\forall~v\in Y~s.t.~g(v)\geqslant1,~\exists~u\in Y~s.t.~g(u)=g(v)-1,~f(v)-f(u)\leqslant w.\end{equation}
  \item $X$ is a nonempty set, and $\sigma:X\longrightarrow Y$ is a surjective map.
\end{itemize}

\setlength{\parindent}{2em}
Now we define $\mu:X\longrightarrow[0,m]$ as
\begin{equation}\mu(t)=m-g(\sigma(t)),\end{equation}
and define $h:X\longrightarrow\mathds{Z}$ as
\begin{equation}h(t)=f(\sigma(t))+w \mu(t)-k.\end{equation}

The following lemma, whose proof is straightforward and thus omitted, lists some basic properties of the functions defined as above.

\setlength{\parindent}{0em}
\begin{lemma}
{
{\bf{(1)}}\,\,For any $y\in Y$ with $g(y)=m$, it holds that $f(y)=k$.

{\bf{(2)}}\,\,For any $v\in Y$, we have $f(v)\leqslant wg(v)$. In particular, we have $k\leqslant wm$.

{\bf{(3)}}\,\,$\mu:X\longrightarrow[0,m]$ is a surjective map.

{\bf{(4)}}\,\,$h[X]\subseteq[0,wm-k]$. Furthermore, for any $t\in X$ with $\mu(t)=m$, it holds that $h(t)=wm-k$.

{\bf{(5)}}\,\,For any $c\in X$ with $\mu(c)\leqslant m-1$, there exists $d\in X$ such that $\mu(d)=\mu(c)+1$, $h(c)\leqslant h(d)$.
}
\end{lemma}

\setlength{\parindent}{2em}
Now, with respect to $(g,f)$, we define $\varphi:[0,k]\longrightarrow[0,m]$ as
\begin{equation}\varphi(a)=\min\{g(u)\mid u\in Y,~a\leqslant f(u)\},\end{equation}
and define $\psi:[0,m]\longrightarrow [0,k]$ as
\begin{equation}\psi(b)=\max\{f(u)\mid u\in Y,~g(u)\leqslant b\}.\end{equation}
Similarly, with respect to $(\mu,h)$, we define $\tau:[0,wm-k]\longrightarrow [0,m]$ as
\begin{equation}\tau(a)=\min\{\mu(t)\mid t\in X,~a\leqslant h(t)\},\end{equation}
and define $\eta:[0,m]\longrightarrow [0,wm-k]$ as
\begin{equation}\eta(b)=\max\{h(t)\mid t\in X,~\mu(t)\leqslant b\}.\end{equation}

\setlength{\parindent}{0em}
\begin{proposition}
{$\varphi$, $\psi$, $\tau$, $\eta$ are well defined. Moreover, we have:

{\bf{(1)}}\,\,$(\varphi,\psi)$ is a Galois connection between $[0,k]$ and $[0,m]$;

{\bf{(2)}}\,\,For any $b\in [0,m]$, $\psi(b)=\max\{f(u)\mid u\in Y,~g(u)=b\}$;

{\bf{(3)}}\,\,$\psi(0)=0$, and for any $l\in [1,m]$, we have $\psi(l)-\psi(l-1)\leqslant w$;

{\bf{(4)}}\,\,$\varphi(a)\in[1,m]$ for all $a\in[1,k]$, and $\varphi(r)+1\leqslant \varphi(r+w)$ for all $r\in[0,k-w]$;

{\bf{(5)}}\,\,$(\tau,\eta)$ is a Galois connection between $[0,wm-k]$ and $[0,m]$;

{\bf{(6)}}\,\,For any $b\in [0,m]$, $\eta(b)=\max\{h(t)\mid t\in X,~\mu(t)=b\}$.
}
\end{proposition}

\begin{proof}
Since $g:Y\longrightarrow[0,m]$ is surjective and $k\in f[Y]$ by (1) of Lemma 3.1, $\varphi$ and $\psi$ are well defined. By (3) and (4) of Lemma 3.1, we have $wm-k\in h[X]$, and hence $\tau$ and $\eta$ are well defined. To finish the proof, we note that (1) and (5) follow from Lemma 2.2, (2) follows from (3.3), (3) follows from (3.1) and (3.4), (4) follows from (1), (3) and Proposition 2.1, and (6) follows from (5) of Lemma 3.1.
\end{proof}

\setlength{\parindent}{2em}
We are now ready to present and prove our first bridging theorem.

\setlength{\parindent}{0em}
\begin{theorem}
{{\bf{(1)}}\,\,For any $l\in [0,m]$, $\eta(l)=\psi(m-l)+wl-k$.

{\bf{(2)}}\,\,For any $\gamma\in\mathds{Z}$, define $\mathcal{A}_{(\gamma)}=\{\varphi(u)\mid u\in[1,k],~u\equiv\gamma+k~(\bmod~w)\}$, $\mathcal{B}_{(\gamma)}=\{m+1-\tau(v)\mid v\in[1,wm-k],~v\equiv\gamma~(\bmod~w)\}$. Then, we have $\mathcal{A}_{(\gamma)}\cap \mathcal{B}_{(\gamma)}=\emptyset$, $\mathcal{A}_{(\gamma)}\cup \mathcal{B}_{(\gamma)}=[1,m]$.
}
\end{theorem}

\begin{proof}
Consider $l\in [0,m]$. Noticing that $\sigma:X\longrightarrow Y$ is surjective, we conclude that $\{\sigma(t)\mid t\in X,~\mu(t)=l\}=\{u\mid u\in Y,~g(u)=m-l\}$. Thus, by (2) and (6) of Proposition 3.1, we have
\begin{eqnarray*}
\begin{split}
\eta(l)&=\max\{f(\sigma(t))+w\mu(t)-k\mid t\in X,~\mu(t)=l\}\\
&=\max\{f(\sigma(t))+wl-k\mid t\in X,~\mu(t)=l\}\\
&=\max\{f(\sigma(t))\mid t\in X,~\mu(t)=l\}+wl-k\\
&=\max\{f(u)\mid u\in Y,~g(u)=m-l\}+wl-k\\
&=\psi(m-l)+wl-k,
\end{split}
\end{eqnarray*}
which completes the proof of (1). Now with (1) and (1), (3), (5) of Proposition 3.1, (2) immediately follows from Theorem \ref{central}.
\end{proof}

\begin{remark}
To summarize, we start with the tuple $(Y,m,g,w,k,f,X,\sigma)$, from which $(\mu,h)$ is determined as in (3.5) and (3.6). Then, with respect to $(g,f)$ and $(\mu,h)$, we define two Galois connections $(\varphi,\psi)$ and $(\tau,\eta)$ via (3.7)--(3.10), which leads to the Wei-type duality theorem as Theorem 3.1.
\end{remark}

\subsection{The second bridging theorem}
\setlength{\parindent}{2em}
In this subsection, we focus on a special case of Theorem 3.1, and we begin with the following definition.

\setlength{\parindent}{0em}
\begin{definition}
Let $(Y,\curlyeqprec)$ be a poset with the least element $0_{_{Y}}$ and the greatest element $\pi_{_{Y}}$, and let $g:Y\longrightarrow\mathds{N}$. Then, the tuple $((Y,\curlyeqprec),g)$ is said to be an abundance if the following three conditions hold:
\begin{equation}\forall~(x,y)\in Y\times Y,~x\curlyeqprec y\Longrightarrow g(x)\leqslant g(y);\end{equation}
\begin{equation}\forall~u\in Y~s.t.~g(u)\leqslant g(\pi_{_{Y}})-1, \exists~v\in Y~s.t.~(u\curlyeqprec v,~g(v)=g(u)+1);\end{equation}
\begin{equation}\forall~v\in Y~s.t.~g(v)\geqslant1, \exists~u\in Y~s.t.~(u\curlyeqprec v,~g(u)=g(v)-1).\end{equation}
\end{definition}

\begin{remark}
Definition 3.1 is inspired by [24. Proposition 1.1], where it is shown that any finite poset has an abundance of ideals, and we will dwell on this particular case in Section 4.
\end{remark}

\setlength{\parindent}{2em}
Now we fix the following notations:

\setlength{\parindent}{0em}
\begin{itemize}
\item $(Y,\curlyeqprec)$ is a poset with the least element $0_{_{Y}}$ and the greatest element \nolinebreak$\pi_{_{Y}}$.
\item $g:Y\longrightarrow\mathds{N}$ such that $((Y,\curlyeqprec),g)$ is an abundance with $g(\pi_{_{Y}})=m$.
\item $w\in\mathbb{Z}^{+}$, $f:Y\longrightarrow \mathbb{Z}$ such that $f(0_{_{Y}})=0$, $f(\pi_{_{Y}})=k\in\mathbb{N}$, and
$$\mbox{For any $(x,y)\in Y\times Y:x\curlyeqprec y\Longrightarrow0\leqslant f(y)-f(x)\leqslant w(g(y)-g(x))$}.$$
\item $X$ is a set and $\sigma:X\longrightarrow Y$ is a surjective map.
\end{itemize}

\setlength{\parindent}{2em}
By Definition 3.1, it can be readily verified that $g$ is a surjective map from $Y$ to $[0, m]$ with $g(0_{_{Y}})=0$. Furthermore, one can easily check that $f$ is a map from $Y$ to $[0,k]$, and the conditions (3.1)--(3.4) hold true for $g$, $w$, $k$ and $f$. Hence by now, we have the tuple $(Y,m,g,w,k,f,X,\sigma)$ to which Theorem 3.1 can be applied. Therefore as a consequence of Theorem 3.1, we have proved our second bridging theorem, as detailed below.

\setlength{\parindent}{0em}
\begin{theorem}
{Define $\mu:X\longrightarrow[0,m]$ as $\mu(t)=m-g(\sigma(t))$, $h:X\longrightarrow\mathds{Z}$ as $h(t)=f(\sigma(t))+w \mu(t)-k$. Furthermore, with respect to $(g,f)$ and $(\mu,h)$, define $\varphi$, $\psi$, $\tau$, $\eta$ exactly in the way as in (3.7)--(3.10). Then, it holds that:

{\bf{(1)}}\,\,For any $l\in [0,m]$, $\eta(l)=\psi(m-l)+wl-k$;

{\bf{(2)}}\,\,For any $\gamma\in\mathds{Z}$, define $\mathcal{A}_{(\gamma)}=\{\varphi(u)\mid u\in[1,k],~u\equiv\gamma+k~(\bmod~w)\}$, $\mathcal{B}_{(\gamma)}=\{m+1-\tau(v)\mid v\in[1,wm-k],~v\equiv\gamma~(\bmod~w)\}$. Then, we have $\mathcal{A}_{(\gamma)}\cap \mathcal{B}_{(\gamma)}=\emptyset$, $\mathcal{A}_{(\gamma)}\cup \mathcal{B}_{(\gamma)}=[1,m]$.
}
\end{theorem}

\setlength{\parindent}{2em}
As an application of Theorem 3.2, in the following example, we derive a Wei-type duality theorem for the $q$-analogue of a matroid ($q$-matroid). We consider $q$-matroids defined over an arbitrary complemented modular lattice of finite length (see [26, Section 10], [12, Definitions 4.4 and 4.13]), which include both matroids and $q$-matroids defined over finite dimensional vector spaces over a field as special cases.

\begin{example}
Let $(Y,\wedge,\vee,\curlyeqprec)$ be a complemented modular lattice of finite length with the least element $0_{_{Y}}$ and the greatest element $\pi_{_{Y}}$. For any $v\in Y$, let $\len(v)$ denote the largest cardinality of a chain in $(Y,\curlyeqprec)$ containing $v$ as its greatest element. Apparently, $\len$ is a map from $Y$ to $\mathds{N}$, and furthermore, it is straightforward to verify that $((Y,\curlyeqprec),\len)$ is an abundance. We let $m=\len(\pi_{_{Y}})$.

Let $\rho:Y\longrightarrow \mathbb{Z}$ such that $((Y,\curlyeqprec),\rho)$ is a $q$-matroid, i.e., the following three conditions hold:

{\bf{(1)}}\,\,For any $v\in Y$, $0\leqslant \rho(v)\leqslant \len(v)$;

{\bf{(2)}}\,\,For any $(x,y)\in Y\times Y$ with $x\curlyeqprec y$, it holds that $\rho(x)\leqslant \rho(y)$;

{\bf{(3)}}\,\,For any $(x,y)\in Y\times Y$, $\rho(x\wedge y)+\rho(x\vee y)\leqslant \rho(x)+\rho(y)$.\\
(1)--(3) imply that $\rho(0_{_{Y}})=0$, and for any $(x,y)\in Y\times Y$ with $x\curlyeqprec y$, it holds true that $0\leqslant \rho(y)-\rho(x)\leqslant \len(y)-\len(x)$. We let $k=\rho(\pi_{_{Y}})$.

Now let $\sigma:Y\longrightarrow Y$ be any bijective map such that
\begin{equation}\forall~(x,y)\in Y\times Y,~\sigma(x)\curlyeqprec\sigma(y)\Longleftrightarrow y\curlyeqprec x.\end{equation}
By (3.14), we deduce that $\len(v)=m-\len(\sigma(v))$ for all $v\in Y$. Define $\theta:Y\longrightarrow \mathbb{Z}$ as $\theta(v)=\rho(\sigma(v))+\len(v)-k$. It can be readily verified that $((Y,\curlyeqprec),\theta)$ is a $q$-matroid, which we will refer to as the dual $q$-matroid of $((Y,\curlyeqprec),\rho)$ with respect to $\sigma$ (c.f. [26, Definition 7.1]).

Now with respect to $(\len,\rho)$ and $(\len,\theta)$, we define $\varphi:[0,k]\longrightarrow[0,m]$, $\psi:[0,m]\longrightarrow [0,k]$, $\tau:[0,m-k]\longrightarrow [0,m]$ and $\eta:[0,m]\longrightarrow [0,m-k]$ as
$$\varphi(a)=\min\{\len(u)\mid u\in Y,~a\leqslant \rho(u)\},$$
$$\psi(b)=\max\{\rho(u)\mid u\in Y,~\len(u)\leqslant b\},$$
$$\tau(a)=\min\{\len(u)\mid u\in Y,~a\leqslant \theta(u)\},$$
$$\eta(b)=\max\{\theta(u)\mid u\in Y,~\len(u)\leqslant b\}.$$
By Theorem 3.2, we conclude that $\eta(l)=\psi(m-l)+l-k$ for all $l\in[0,m]$. Furthermore, $\{\varphi(u)\mid u\in[1,k]\}$ and $\{m+1-\tau(v)\mid v\in[1,m-k]\}$ form a partition of $[1,m]$. The latter result may be regarded as the Wei-type duality theorem for $q$-matroids, which also generalizes the Wei-type duality theorems for matroids [6, Theorem 1], for sum-matroids [35, Theorem 12] and for sum-rank metric codes [30, Theorem 2].
\end{example}

\section{Wei-type duality theorems for $w$-demi-matroids}
\setlength{\parindent}{2em}

We begin with a brief overview of the notion of demi-matroids. Demi-matroids are first introduced by Britz, Johnsen, Mayhew and Shiromoto in \cite{6} as a natural generalization of matroids (see \cite{33}). In \cite{6}, two fundamental Wei-type duality theorems for demi-matroids are established. Later in \cite{5}, Britz, Johnsen and Martin show that demi-matroids naturally arise from code flags of linear codes over division rings. Furthermore, it has been proven that the correspondence between a code flag and its dual is represented by demi-matroid duality in a natural way, and generalized weights and profiles of code flags can be recast as those of their associated demi-matroids (see [5, Sections II, III] for more details). In this section, we prove two Wei-type duality theorems for $w$-demi-matroids, a notion which generalizes the notion of demi-matroids. And in Sections 6, we will show that similar to the demi-matroid case, $w$-demi-matroids arise naturally from code flags over modules.

Throughout this section, we let $E$ be a finite set with $|E|=m$.

\setlength{\parindent}{0em}
\begin{definition}
For any $f:2^{E}\longrightarrow\mathds{Z}$ and $w\in\mathds{Z}^{+}$, $(E,f)$ is said to be a $w$-demi-matroid if the following two conditions hold:

{\bf{(1)}}\,\,$f(\emptyset)=0$;

{\bf{(2)}}\,\,For any $A,B\subseteq E$ with $A\subseteq B$, $0\leqslant f(B)-f(A)\leqslant w(|B|-|A|)$.
\end{definition}

\setlength{\parindent}{2em}
\begin{remark}
The notion of a $w$-demi-matroid generalizes that of the well known demi-matroid \cite{5,6} as a demi-matroid is a de facto $1$-demi-matroid.
\end{remark}

\setlength{\parindent}{2em}
The proof of the following proposition is straightforward and thus omitted.

\setlength{\parindent}{0em}
\begin{proposition}
{For $w\in\mathds{Z}^{+}$ and a $w$-demi-matroid $(E,f)$, define $h:2^{E}\longrightarrow\mathds{Z}$ as
$$h(A)=f(E-A)+w|A|-f(E).$$
Then, $(E,h)$ is a $w$-demi-matroid with $h(E)=wm-f(E)$. Furthermore, for any $A\subseteq E$, it holds that $f(A)=h(E-A)+w|A|-h(E)$.
}
\end{proposition}

\setlength{\parindent}{2em}
\begin{remark}
Similar to the demi-matroid case (\cite{5,6,29}), $(E,h)$ defined in Proposition 4.1 can be regarded as the dual $w$-demi-matroid of $(E,f)$. Proposition 4.1 implies that $(E,f)$ is also the dual $w$-demi-matroid of $(E,h)$; in other words, the dual operation for $w$-demi-matroid is an involution (see [6, Page 5]).
\end{remark}

Now we proceed to define generalized weights and profiles for a $w$-demi-matroid. First of all, we consider a special case of Definition 3.1. For $\mathcal{C}\subseteq2^{E}$ such that $\emptyset,E\in\mathcal{C}$, we define $g:\mathcal{C}\longrightarrow\mathds{N}$ as $g(A)=|A|$. Then, by Definition 3.1, $((\mathcal{C},\subseteq),g)$ is an abundance if and only if the following two conditions hold:

\setlength{\parindent}{0em}
{\bf{(1)}}\,\,For any $A\in \mathcal{C}$ with $A\subsetneqq E$, there exists $B\in\mathcal{C}$ such that $A\subseteq B$, $|B|=|A|+1$;

{\bf{(2)}}\,\,For any $B\in \mathcal{C}$ with $B\neq\emptyset$, there exists $A\in\mathcal{C}$ such that $A\subseteq B$, $|A|=|B|-1$.

For convenience, throughout this section, we will simply say that \textit{$\mathcal{C}$ is an abundance} if $((\mathcal{C},\subseteq),g)$ is an abundance. Now consider $\mathcal{D}\subseteq2^{E}$ defined as
$$\mathcal{D}=\{E-A\mid A\in\mathcal{C}\}.$$
Then, it is straightforward to verify that $\emptyset,E\in\mathcal{D}$. Furthermore, $\mathcal{C}$ is an abundance if and only if $\mathcal{D}$ is an abundance.

\setlength{\parindent}{2em}
Now we define generalized weights and profiles for $w$-demi-matroids. Throughout the rest of this section, we let $w$ be a fixed positive integer.

\setlength{\parindent}{0em}
\begin{definition}
{Let $(E,f)$ be a $w$-demi-matroid with $k=f(E)$, and let $\mathcal{C}\subseteq2^{E}$ such that $\emptyset,E\in\mathcal{C}$, $\mathcal{C}$ is an abundance. Then, for any $a\in[0,k]$, the $a$-th generalized weight of $(f,\mathcal{C})$, denoted by $\mathbf{d}_{a}(f,\mathcal{C})$, is defined as
$$\mathbf{d}_{a}(f,\mathcal{C})=\min\{|B|\mid B\in\mathcal{C},~a\leqslant f(B)\},$$
and for any $b\in[0,m]$, the $b$-th profile of $(f,\mathcal{C})$, denoted by $\mathbf{K}_{b}(f,\mathcal{C})$, is defined as
$$\mathbf{K}_{b}(f,\mathcal{C})=\max\{f(B)\mid B\in\mathcal{C},~|B|=b\}.$$
}
\end{definition}

\setlength{\parindent}{2em}
Now we state and prove our first Wei-type duality theorem for $w$-demi-matroids.

\setlength{\parindent}{0em}
\begin{theorem}
{Let $(E,f)$ be a $w$-demi-matroid with $k=f(E)$, and define $h:2^{E}\longrightarrow\mathds{Z}$ as $h(B)=f(E-B)+w|B|-k$. Furthermore, let $\mathcal{C}\subseteq2^{E}$ such that $\emptyset,E\in\mathcal{C}$, $\mathcal{C}$ is an abundance, and let $\mathcal{D}=\{E-A\mid A\in\mathcal{C}\}$. Then, it holds that:

{\bf{(1)}}\,\,For any $l\in [0,m]$, $\mathbf{K}_{l}(h,\mathcal{D})=\mathbf{K}_{m-l}(f,\mathcal{C})+wl-k$;

{\bf{(2)}}\,\,For any $\gamma\in\mathds{Z}$, let $\mathcal{A}_{(\gamma)}=\{\mathbf{d}_{u}(f,\mathcal{C})\mid u\in[1,k],~u\equiv\gamma+k~(\bmod~w)\}$, $\mathcal{B}_{(\gamma)}=\{m+1-\mathbf{d}_{v}(h,\mathcal{D})\mid v\in[1,wm-k],~v\equiv\gamma~(\bmod~w)\}$. Then, we have $\mathcal{A}_{(\gamma)}\cap \mathcal{B}_{(\gamma)}=\emptyset$, $\mathcal{A}_{(\gamma)}\cup \mathcal{B}_{(\gamma)}=[1,m]$.
}
\end{theorem}

\begin{proof}
Define $g:\mathcal{C}\longrightarrow\mathds{Z}$ as $g(A)=|A|$. Then, $((\mathcal{C},\subseteq),g)$ is an abundance with $g(E)=m$. Consider $f\mid_{\mathcal{C}}:\mathcal{C}\longrightarrow\mathbb{Z}$. Then, we have $f\mid_{\mathcal{C}}(E)=f(E)=k$. Furthermore, since $(E,f)$ is a $w$-demi-matroid, by Definition 4.1, $f\mid_{\mathcal{C}}$ satisfies the following two conditions:

$(i)$\,\,$f\mid_{\mathcal{C}}(\emptyset)=0$;

$(ii)$\,\,For any $A,B\in \mathcal{C}$ with $A\subseteq B$, $0\leqslant f\mid_{\mathcal{C}}(B)-f\mid_{\mathcal{C}}(A)\leqslant w(g(B)-g(A))$.
Define $\sigma:\mathcal{D}\longrightarrow \mathcal{C}$ as $\sigma(B)=E-B$. Obviously, $\sigma$ is a bijective map from $\mathcal{D}$ to $\mathcal{C}$. Now, define $\mu:\mathcal{D}\longrightarrow\mathds{Z}$ as $\mu(B)=m-g(\sigma(B))$, and consider $h\mid_{\mathcal{D}}:\mathcal{D}\longrightarrow\mathbb{Z}$. Then, for any $B\in \mathcal{D}$, it can be readily verified that
$$\text{$\mu(B)=|B|$ and $h\mid_{\mathcal{D}}(B)=f\mid_{\mathcal{C}}(\sigma(B))+w\cdot\mu(B)-k$}.$$
Now, with respect to $(g,f\mid_{\mathcal{C}})$ and $(\mu,h\mid_{\mathcal{D}})$, we define $\varphi$, $\psi$, $\tau$, $\eta$ exactly in the way as in (3.7)--(3.10). Then, by Definition 4.2, for any $a\in[0,k]$, $b\in[0,m]$, we have $\varphi(a)=\mathbf{d}_{a}(f,\mathcal{C})$, $\psi(b)=\mathbf{K}_{b}(f,\mathcal{C})$. Similarly, for any $a\in[0,wm-k]$, $b\in [0,m]$, it holds that $\tau(a)=\mathbf{d}_{a}(h,\mathcal{D})$, $\eta(b)=\mathbf{K}_{b}(h,\mathcal{D})$. Therefore, an application of Theorem 3.2 immediately leads to the desired result.
\end{proof}

\setlength{\parindent}{2em}
From now on, we focus our attention to posets over $E$, and we begin with some notations. For a poset $\mathbf{P}=(E,\preccurlyeq_{\mathbf{P}})$, recall that for any $B\subseteq E$, $B$ is referred to as an \textit{ideal} of $\mathbf{P}$ if for any $v\in B$ and $u\in E$, $u\preccurlyeq_{\mathbf{P}}v$ implies $u\in B$. We let $\mathcal{I}(\mathbf{P})$ denote the set of all the ideals of $\mathbf{P}$. Furthermore, the \textit{dual poset} of $\mathbf{P}$, denoted by $\mathbf{\overline{P}}=(E,\preccurlyeq_{\mathbf{\overline{P}}})$, is defined as
$$\text{$u\preccurlyeq_{\mathbf{\overline{P}}} v\Longleftrightarrow v\preccurlyeq_{\mathbf{P}}u$ for all $(u,v)\in E\times E$}.$$

The relation between posets over $E$ and the aforementioned notion of abundance is detailed in the following proposition.

\setlength{\parindent}{0em}
\begin{proposition}
{\bf{(1)}}\,\,Let $\mathbf{P}=(E,\preccurlyeq_{\mathbf{P}})$ be a poset. Then $\emptyset,E\in\mathcal{I}(\mathbf{P})$, $\mathcal{I}(\mathbf{P})$ is an abundance, and for any $I,J\in\mathcal{I}(\mathbf{P})$, we have $I\cap J\in\mathcal{I}(\mathbf{P})$, $I\cup J\in\mathcal{I}(\mathbf{P})$.

{\bf{(2)}}\,\,Let $\mathcal{C}\subseteq 2^{E}$ such that $\emptyset,E\in\mathcal{C}$, $\mathcal{C}$ is an abundance, and for any $I,J\in\mathcal{C}$, it holds that $I\cap J\in\mathcal{C}$, $I\cup J\in\mathcal{C}$. Then, there exists a poset $\mathbf{Q}=(E,\preccurlyeq_{\mathbf{Q}})$ such that $\mathcal{C}=\mathcal{I}(\mathbf{Q})$.

{\bf{(3)}}\,\,Suppose that $E=\{1,2,3\}$ and $\mathcal{C}=\{\emptyset,\{1\},\{2\},\{1,3\},\{2,3\},\{1,2,3\}\}$. Then, we have $\emptyset,E\in\mathcal{C}$, and $\mathcal{C}$ is an abundance. Furthermore, for any poset $\mathbf{P}=(E,\preccurlyeq_{\mathbf{P}})$, we have $\mathcal{C}\neq\mathcal{I}(\mathbf{P})$.
\end{proposition}

\begin{proof}
(1)\,\,By [32, Proposition 7] or [24, Proposition 1.1], we deduce that $\mathcal{I}(\mathbf{P})$ is an abundance. Now the rest is well-known and is also straightforward to verify.

(2)\,\,Since $E\in\mathcal{C}$ and $I\cap J\in\mathcal{C}$ for all $I,J\in\mathcal{C}$, for any $e\in E$, there uniquely exists $H_{(e)}\in \mathcal{C}$ such that $e\in H_{(e)}$, and
$$\text{For any $B\in\mathcal{C}$, $e\in B\Longleftrightarrow H_{(e)}\subseteq B$}.$$
Now define the relation $\preccurlyeq$ on $E$ such that
$$\text{$u\preccurlyeq v\Longleftrightarrow H_{(u)}\subseteq H_{(v)}$ for all $u,v\in E$}.$$
Since $\mathcal{C}$ is an abundance, it is routine to check that $\preccurlyeq$ is indeed a partial order. Furthermore, by $I\cup J\in\mathcal{C}$ for all $I,J\in\mathcal{C}$, one can check that $\mathcal{C}=\mathcal{I}(E,\preccurlyeq)$. We omit the details of the verification.

(3)\,\,It can be readily verified that $\emptyset,E\in\mathcal{C}$ and $\mathcal{C}$ is an abundance. Noticing that $\{1\}\cup\{2\}=\{1,2\}\not\in\mathcal{C}$, by (1), we conclude that $\mathcal{C}\neq\mathcal{I}(\mathbf{P})$ for any poset $\mathbf{P}=(E,\preccurlyeq_{\mathbf{P}})$.
\end{proof}

\setlength{\parindent}{2em}
Based on Proposition 4.2, as a special case of Definition 4.2, we give the following Definition.

\setlength{\parindent}{0em}
\begin{definition}
{Let $(E,f)$ be a $w$-demi-matroid with $k=f(E)$, and suppose that $\mathbf{P}=(E,\preccurlyeq_{\mathbf{P}})$ is a poset. Then, for any $a\in[0,k]$, the $a$-th generalized weight of $(f,\mathbf{P})$, denoted by $\mathbf{d}_{a}(f,\mathbf{P})$, is defined as
$$\mathbf{d}_{a}(f,\mathbf{P})=\mathbf{d}_{a}(f,\mathcal{I}(\mathbf{P}))=\min\{|B|\mid B\in\mathcal{I}(\mathbf{P}),~a\leqslant f(B)\},$$
and for any $b\in[0,m]$, the $b$-th profile of $(f,\mathbf{P})$, denoted by $\mathbf{K}_{b}(f,\mathbf{P})$, is defined as
$$\mathbf{K}_{b}(f,\mathbf{P})=\mathbf{K}_{b}(f,\mathcal{I}(\mathbf{P}))=\max\{f(B)\mid B\in\mathcal{I}(\mathbf{P}),~|B|=b\}.$$
}
\end{definition}

\setlength{\parindent}{2em}
\begin{remark}
{Let $\mathbf{P}=(E,\preccurlyeq_{\mathbf{P}})$ be a poset, and for any $B\subseteq E$, let $\langle B\rangle_{\mathbf{P}}=\{u\mid u\in E,~\exists~v\in B~s.t.~u\preccurlyeq_{\mathbf{P}}v\}$ denote the ideal generated by $B$. Then, for any $a\in[0,k]$, $\mathbf{d}_{a}(f,\mathbf{P})$ can be alternatively defined as
$$\min\{|\langle B\rangle_{\mathbf{P}}|\mid B\subseteq E,~a\leqslant f(B)\},$$
and for any $b\in[0,m]$, $\mathbf{K}_{b}(f,\mathbf{P})$ can be alternatively defined as
$$\max\{f(B)\mid B\subseteq E,~|\langle B\rangle_{\mathbf{P}}|\leqslant b\}.$$
}
\end{remark}

\setlength{\parindent}{2em}
Now we state and prove our second Wei-type duality theorem for $w$-demi-matroids.

\setlength{\parindent}{0em}
\begin{theorem}
{Let $(E,f)$ be a $w$-demi-matroid with $k=f(E)$, and define $h:2^{E}\longrightarrow\mathds{Z}$ as $h(B)=f(E-B)+w|B|-k$. Then, for any poset $\mathbf{P}=(E,\preccurlyeq_{\mathbf{P}})$, we have:

{\bf{(1)}}\,\,For any $l\in [0,m]$, $\mathbf{K}_{l}(h,\mathbf{\overline{P}})=\mathbf{K}_{m-l}(f,\mathbf{P})+wl-k$;

{\bf{(2)}}\,\,For any $\gamma\in\mathds{Z}$, let $\mathcal{A}_{(\gamma)}=\{\mathbf{d}_{u}(f,\mathbf{P})\mid u\in[1,k],~u\equiv\gamma+k~(\bmod~w)\}$, $\mathcal{B}_{(\gamma)}=\{m+1-\mathbf{d}_{v}(h,\mathbf{\overline{P}})\mid v\in[1,wm-k],~v\equiv\gamma~(\bmod~w)\}$. Then, we have $\mathcal{A}_{(\gamma)}\cap \mathcal{B}_{(\gamma)}=\emptyset$, $\mathcal{A}_{(\gamma)}\cup \mathcal{B}_{(\gamma)}=[1,m]$.
}
\end{theorem}

\begin{proof}
Let $\mathbf{P}=(E,\preccurlyeq_{\mathbf{P}})$ be a poset. Then, by [24, Lemma 1.2], we deduce that $\mathcal{I}(\mathbf{\overline{P}})=\{E-B\mid B\in\mathcal{I}(\mathbf{P})\}$. Now the result immediately follows from Definition 4.3 and Theorem 4.1.
\end{proof}

\setlength{\parindent}{2em}
\begin{remark}
As in Theorem 4.2, with $w$ set to be $1$, we recover the Wei-type duality theorem for demi-matroids [6, Theorem 6]. If we let $\preccurlyeq_{\mathbf{P}}$ be the trivial (anti-chain) partial order on $E$, then Theorem 4.2 becomes an analogue of the Wei-type duality theorem for $(q,m)$-demi-polymatroids [7, Theorem 4].
\end{remark}

\section{Wei-type duality theorems for $w$-demi-polymatroids}
\setlength{\parindent}{2em}
As combinatorial notions arising naturally from rank metric codes, $(q,m)$-polymatroids have been independently introduced by Shiromoto in \cite{38} and by Gorla, Jurrius, l\'{o}pez and Ravagnani in \cite{21}, which have been further generalized to the notion of $(q,m)$-demi-polymatroids by Britz, Mammomiti and Shiromotoin \cite{7}, and by Ghorpade and Johnsen in \cite{18}. $(q,m)$-demi-polymatroids defined over vector spaces over a field have been studied extensively in \cite{7} and \cite{18}, and in both papers, Wei-type duality theorems for $(q,m)$-demi-polymatroids have been established. In this section, we propose and study $w$-demi-polymatroids defined over modules with a composition series, and prove the corresponding Wei-type duality theorems. And in Section 6, we will show that $w$-demi-polymatroids arise naturally from codes over modules.

Throughout this section, we let $R$ and $S$ be rings, which are always assumed to be associative with multiplicative identity. Recall that for any left $R$-module $M$, $M$ is said to be \textit{simple} if $M\neq\{0\}$ and $M$ has no left $R$-submodules other than $\{0\}$ and $M$. A \textit{composition series} of $M$ is a chain of left $R$-submodules $$\{0\}=L_{0}\subseteq L_{1}\subseteq\cdots\subseteq L_{t-1}\subseteq L_{t}=M,$$
where $t\in\mathds{N}$ and $L_{i}/L_{i-1}$ is a simple left $R$-module for all $i\in[1,t]$ (see, e.g., [1, Chapter 3], [11, Section 13], [39, Section 4]). For any left $R$-module $M$, $\mathcal{P}(M)$ will denote the set of all the left $R$-submodules of $M$. We note that all the above notions parallelly apply to right $S$-modules.

Now we present a notion that is a slight modified version of [39, Section 4, Paragraph 4]. Let $\Omega$ be a nonempty collection of simple left $R$-modules. For any left $R$-module $M$ with a composition series $\{0\}=L_{0}\subseteq L_{1}\subseteq\cdots\subseteq L_{t}=M$, we define $\rho_{\Omega}(M)\in\mathbb{N}$ as
\begin{equation}\hspace*{-4mm}\rho_{\Omega}(M)=|\{i\mid i\in[1,t], \exists~W\in\Omega~s.t.~L_{i}/L_{i-1}\cong W~\text{as left $R$-modules}\}|.\end{equation}
Similarly, let $\Delta$ be a nonempty collection of simple right $S$-modules. For any right $S$-module $N$ with a composition series $\{0\}=H_{0}\subseteq H_{1}\subseteq\cdots\subseteq H_{p-1}\subseteq H_{p}=N$, we define $\lambda_{\Delta}(N)\in\mathbb{N}$ as
\begin{equation}\hspace*{-5mm}\lambda_{\Delta}(N)=|\{i\mid i\in[1,p],\exists~Z\in\Delta~s.t.~H_{i}/H_{i-1}\cong Z~\text{as right $S$-modules}\}|.\end{equation}

By the Jordan-H\"{o}lder theorem (see [11, Theorem 13.7]), both $\rho_{\Omega}(M)$ and $\lambda_{\Delta}(N)$ are independent of the choice of the composition series, and hence well-defined. We will show in Section 5.2 that $\rho_{\Omega}$ and $\lambda_{\Delta}$ are closely related in suitable settings, which will lead to Wei-type duality theorems.

\begin{remark}
If $R$ is a field, which implies that the only simple left $R$-module up to isomorphism is $_{R}R$, then for any finite dimensional vector space $M$ over $R$, we have $\rho_{\Omega}(M)=\dim_{R}(M)$. Therefore the notion $\rho_{\Omega}$ naturally extends the dimension function of vector spaces over a field. More generally, let $R$ be an arbitrary ring, as long as $\Omega$ contains all the simple left $R$-modules up to isomorphism, $\rho_{\Omega}(M)$ then becomes the length of a composition series of $M$. We note that similar properties hold for $\lambda_{\Delta}$ in a parallel fashion.
\end{remark}

\subsection{Generalized weights and profiles of $w$-demi-polymatroids}

In this subsection, we will focus on left $R$-modules, and all the discussions remain valid for right $S$-modules parallelly. Throughout, we let $\Omega$ be a nonempty collection of simple left $R$-modules, and let $M$ be a left $R$-module with a composition series.

As detailed in the following Lemma, which is a consequence of [39, Proposition 4.1], $\rho_{\Omega}$ maintains some basic properties of the dimension function of vector spaces.

\setlength{\parindent}{0em}
\begin{lemma}
{{\bf{(1)}}\,\,Define $g:\mathcal{P}(M)\longrightarrow\mathds{N}$ as $g(V)=\rho_{\Omega}(V)$. Then, $((\mathcal{P}(M),\subseteq),g)$ is an abundance.

{\bf{(2)}}\,\,For any $U,V\in\mathcal{P}(M)$ with $U\subseteq V$, $\rho_{\Omega}(V)=\rho_{\Omega}(U)+\nolinebreak\rho_{\Omega}(V/U)$.

{\bf{(3)}}\,\,For any $X,Y\in\mathcal{P}(M)$, $\rho_{\Omega}(X+Y)+\rho_{\Omega}(X\cap Y)=\rho_{\Omega}(X)+\rho_{\Omega}(Y)$.

{\bf{(4)}}\,\,For $s\in\mathds{Z}^{+}$, $(L_{(1)},\dots,L_{(s)})\in\mathcal{P}(M)^{s}$ with $L_{(s)}\subseteq\cdots\subseteq L_{(1)}$, and $X,Y\in\mathcal{P}(M)$ with $X\subseteq Y$, it holds that
$$\hspace*{-6mm}\mbox{$0\leqslant\left(\sum_{i=1}^{s}(-1)^{i-1}\rho_{\Omega}(L_{(i)}\cap Y)\right)-\left(\sum_{i=1}^{s}(-1)^{i-1}\rho_{\Omega}(L_{(i)}\cap X)\right)\leqslant\rho_{\Omega}(Y)-\rho_{\Omega}(X)$}.$$
}
\end{lemma}

\setlength{\parindent}{2em}
Now we define $w$-demi-polymatroids over modules.

\setlength{\parindent}{0em}
\begin{definition}
{For any $f:\mathcal{P}(M)\longrightarrow\mathds{Z}$ and $w\in\mathds{Z}^{+}$, $(M,f,\Omega)$ is said to be a $w$-demi-polymatroid if the following two conditions hold:

{\bf{(1)}}\,\,$f(\{0\})=0$;

{\bf{(2)}}\,\,For any $X,Y\in\mathcal{P}(M)$ with $X\subseteq Y$, $0\leqslant f(Y)-f(X)\leqslant w(\rho_{\Omega}(Y)-\nolinebreak\rho_{\Omega}(X))$.
}
\end{definition}

\setlength{\parindent}{2em}
We note that if $R$ is field in Definition 5.1, then $\rho_{\Omega}=\dim_{R}$ by Remark 5.1, and we recover the definition of $(q,m)$-demi-polymatroids defined over vector spaces (see [18, Definition 1], [7, Section 3]).

Now we define generalized weights and profiles for $w$-demi-polymatroids.

\setlength{\parindent}{0em}
\begin{definition}
{Let $f:\mathcal{P}(M)\longrightarrow\mathds{Z}$ and $w\in\mathds{Z}^{+}$ such that $(M,f,\Omega)$ is a $w$-demi-polymatroid with $f(M)=k$, and let $\Phi\subseteq\mathcal{P}(M)$ such that $\{0\},M\in\Phi$ and $((\Phi,\subseteq),(\rho_{\Omega})\mid_{\Phi})$ is an abundance. For any $a\in[0,k]$, the $a$-th generalized weight of $((M,f,\Omega),\Phi)$, denoted by $\mathbf{d}_{a}((M,f,\Omega),\Phi)$, is defined as $$\mathbf{d}_{a}((M,f,\Omega),\Phi)=\min\{\rho_{\Omega}(W)\mid W\in \Phi,~a\leqslant f(W)\},$$
and for any $b\in[0,\rho_{\Omega}(M)]$, the $b$-th profile of $((M,f,\Omega),\Phi)$, denoted by $\mathbf{K}_{b}((M,f,\Omega),\Phi)$, is defined as
$$\mathbf{K}_{b}((M,f,\Omega),\Phi)=\max\{f(W)\mid W\in \Phi,~\rho_{\Omega}(W)=b\}.$$
}
\end{definition}

\setlength{\parindent}{2em}
Adopting the notations in Definition 5.2, an application of Proposition 3.1 yields the following corollary.

\setlength{\parindent}{0em}
\begin{corollary}
{\bf{(1)}}\,\,Define $\varphi:[0,k]\longrightarrow[0,\rho_{\Omega}(M)]$ as
$$\varphi(a)=\mathbf{d}_{a}((M,f,\Omega),\Phi),$$
and define $\psi:[0,\rho_{\Omega}(M)]\longrightarrow[0,k]$ as
$$\psi(b)=\mathbf{K}_{b}((M,f,\Omega),\Phi).$$
Then, $(\varphi,\psi)$ is a Galois connection between $[0,k]$ and $[0,\rho_{\Omega}(M)]$.

{\bf{(2)}}\,\,$\mathbf{K}_{0}((M,f,\Omega),\Phi)=0$.

{\bf{(3)}}\,\,$\mathbf{K}_{l}((M,f,\Omega),\Phi)-\mathbf{K}_{l-1}((M,f,\Omega),\Phi)\leqslant w$ for all $l\in[1,\rho_{\Omega}(M)]$.

{\bf{(4)}}\,\,$\mathbf{d}_{r}((M,f,\Omega),\Phi)+1\leqslant\mathbf{d}_{r+w}((M,f,\Omega),\Phi)$ for all $r\in[0,k-w]$.
\end{corollary}

\subsection{Wei-type duality theorems}
\setlength{\parindent}{2em}
A necessary step for establishing Wei-type duality theorems for $w$-demi-polymatroids is to extend the notion of dual space of vector space over a field. To this end, we will use non-degenerated bilinear maps defined for modules (see [1, Theorem 30.1] and [11, Theorem 58.8]). We begin with the following definition, where we follow [1, Chapter 6, Section 24].

\setlength{\parindent}{0em}
\begin{definition}
Let $M$ be a left $R$-module, $N$ be a right $S$-module, and $U$ be an $R$-$S$ bimodule. Fix $\varpi:M\times N\longrightarrow U$. Then, $\varpi$ is said to be a bilinear map if for any $a,b\in M$, $c,d\in N$, $r\in R$, $s\in S$, it holds that

{\bf{(1)}}\,\,$\varpi(a+b,c)=\varpi(a,c)+\varpi(b,c)$;

{\bf{(2)}}\,\,$\varpi(a,c+d)=\varpi(a,c)+\varpi(a,d)$;

{\bf{(3)}}\,\,$\varpi(ra,cs)=r\varpi(a,c)s$.

Assume in addition that $\varpi$ is a bilinear map. Then, for any $C\subseteq M$, the \textit{right annihilator of $C$ with respect to $\varpi$}, denoted by $C^{\bot}$, is defined as
\begin{equation}C^{\bot}=\left\{y\mid y\in N,~\varpi(x,y)=0~\text{for all}~x\in C\right\},\end{equation}
and for any $D\subseteq N$, the \textit{left annihilator of $D$ with respect to $\varpi$}, denoted by $^{\bot}D$, is defined as
\begin{equation}^{\bot}D=\left\{x\mid x\in M,~\varpi(x,y)=0~\text{for all}~y\in D\right\}.\end{equation}
Finally, $\varpi$ is said to be non-degenerated if $M^{\bot}=\{0\}$, $^{\bot}N=\{0\}$.
\end{definition}

\setlength{\parindent}{2em}
Throughout the rest of this subsection, we fix the following notations:

\setlength{\parindent}{0em}
\begin{itemize}
\item $M$ is a left $R$-module with a composition series, and $N$ is a right $S$-module.
\item $U$ is an $R$-$S$ bimodule satisfying the following two conditions:
\begin{equation}\hspace*{-1mm}\mbox{\begin{footnotesize}For any simple left $R$-module $X$, ${\Hom_{R}}(X,U)$ is a simple right $S$-module;\end{footnotesize}}\end{equation}
\begin{equation}\hspace*{-1mm}\mbox{\begin{footnotesize}For any simple right $S$-module $Y$, ${\Hom_{S}}(Y,U)$ is a simple left $R$-module.\end{footnotesize}}\end{equation}
\item $\varpi:M\times N\longrightarrow U$ is a non-degenerated bilinear map.
\end{itemize}

\setlength{\parindent}{2em}
As a consequence of [1, Theorem 30.1] or [39, Equations (45) and (47)], the right $S$-module $N$ has a composition series. Furthermore, for any $C\in\mathcal{P}(M)$, $D\in\mathcal{P}(N)$, it holds that \begin{equation}^{\bot}(C^{\bot})=C,~(^{\bot}D)^{\bot}=D.\end{equation}

From now on, the following two notations are also fixed:

\setlength{\parindent}{0em}
\begin{itemize}
\item $\Omega$ is a nonempty collection of simple left $R$-modules.
\item $\Delta=\{{\Hom_{R}}(X,U)\mid X\in \Omega\}$.
\end{itemize}

\setlength{\parindent}{2em}
By (5.5), $\Delta$ is a collection of simple right $S$-modules. Furthermore, as a consequence of [1, Theorem 30.1], $\rho_{\Omega}$ and $\lambda_{\Delta}$ is connected by the following equation:
\begin{equation}\mbox{$\lambda_{\Delta}(D)=\rho_{\Omega}(M)-\rho_{\Omega}(^{\bot}D)$ for all $D\in\mathcal{P}(N)$}.\end{equation}

\setlength{\parindent}{2em}
Now we are ready to state and prove Wei-type duality theorems for $w$-demi-polymatroids. Fix $f:\mathcal{P}(M)\longrightarrow\mathds{Z}$ and $w\in\mathds{Z}^{+}$ such that $(M,f,\Omega)$ is a $w$-demi-polymatroid with $f(M)=k$, and define $h:\mathcal{P}(N)\longrightarrow\mathds{Z}$ as
$$h(D)=f(^{\bot}D)+w\cdot\lambda_{\Delta}(D)-k.$$
By (5.7) and (5.8), $(N,h,\Delta)$ is indeed a $w$-demi-polymatroid with $h(N)=w\cdot\rho_{\Omega}(M)-k$. Furthermore, for any $C\in\mathcal{P}(M)$, it holds that
$$f(C)=h(C^{\bot})+w\cdot \rho_{\Omega}(C)-h(N).$$
Similar to the $w$-demi-matroid case Proposition 4.1, $(N,h,\Delta)$ can be regarded as the dual $w$-demi-polymatroid of $(M,f,\Omega)$.

We also fix $\Phi\subseteq\mathcal{P}(M)$ such that $\{0\},M\in\Phi$, $((\Phi,\subseteq),(\rho_{\Omega})\mid_{\Phi})$ is an abundance, and define $\Theta\subseteq\mathcal{P}(N)$ as
$$\Theta=\{X^{\bot}\mid X\in\Phi\}.$$
By (5.7) and (5.8), we deduce that $\{0\},N\in\Theta$, and $((\Theta,\subseteq),(\lambda_{\Delta})\mid_{\Theta})$ is an abundance. Furthermore, we have the following Wei-type duality theorem for $((M,f,\Omega),\Phi)$ and $((N,h,\Delta),\Theta)$.

\setlength{\parindent}{0em}
\begin{theorem}
{{\bf{(1)}}\,\,For any $l\in [0,\rho_{\Omega}(M)]$,
$$\mathbf{K}_{l}((N,h,\Delta),\Theta)=\mathbf{K}_{\rho_{\Omega}(M)-l}((M,f,\Omega),\Phi)+wl-k.$$
{\bf{(2)}}\,\,For any $\gamma\in\mathds{Z}$, define the sets $\mathcal{A}_{(\gamma)}$ and $\mathcal{B}_{(\gamma)}$ as
$$\mathcal{A}_{(\gamma)}=\{\mathbf{d}_{a}((M,f,\Omega),\Phi)\mid a\in[1,k],a\equiv\gamma+k~(\bmod~w)\},$$
$$\mathcal{B}_{(\gamma)}=\{\rho_{\Omega}(M)+1-\mathbf{d}_{c}((N,h,\Delta),\Theta)\mid c\in[1,w\cdot\rho_{\Omega}(M)-k],c\equiv\gamma~(\bmod~w)\}.$$
Then, for any $\gamma\in\mathds{Z}$, we have $\mathcal{A}_{(\gamma)}\cap \mathcal{B}_{(\gamma)}=\emptyset$, $\mathcal{A}_{(\gamma)}\cup \mathcal{B}_{(\gamma)}=[1,\rho_{\Omega}(M)]$.
}
\end{theorem}

\begin{proof}
Let $g=(\rho_{\Omega})\mid_{\Phi}$. Then, $((\Phi,\subseteq),g)$ is an abundance with $g(M)=\rho_{\Omega}(M)$. Consider $f\mid_{\Phi}:\Phi\longrightarrow\mathbb{Z}$. Then, we have $f\mid_{\Phi}(M)=f(M)=k$. Furthermore, since $(M,f,\Omega)$ is a $w$-demi-polymatroid, by Definition 5.1, $f\mid_{\Phi}$ satisfies the following two conditions:

$(i)$\,\,$f\mid_{\Phi}(\{0\})=0$;

$(ii)$\,\,For any $X,Y\in\Phi$ with $X\subseteq Y$, $0\leqslant f\mid_{\Phi}(Y)-f\mid_{\Phi}(X)\leqslant w(g(Y)-\nolinebreak g(X))$.
Moreover, by (5.7), there uniquely exists a bijective map $\sigma:\Theta\longrightarrow\Phi$ such that $\sigma(D)={^{\bot}D}$ for any $D\in\Theta$. Now, define $\mu:\Theta\longrightarrow[0,\rho_{\Omega}(M)]$ as $\mu(D)=\rho_{\Omega}(M)-g(\sigma(D))$, and consider $h\mid_{\Theta}:\Theta\longrightarrow\mathbb{Z}$. Then, by (5.8), for any $D\in\Theta$, it can be readily verified that
$$\text{$\mu(D)=\lambda_{\Delta}(D)$ and $h\mid_{\Theta}(D)=f\mid_{\Phi}(\sigma(D))+w\cdot\mu(D)-k$}.$$
Now, with respect to $(g,f\mid_{\Phi})$ and $(\mu,h\mid_{\Theta})$, we define $\varphi$, $\psi$, $\tau$, $\eta$ exactly in the way as in (3.7)--(3.10). Then, by Definition 5.2, for any $a\in[0,k]$, $b\in[0,\rho_{\Omega}(M)]$, we have $\varphi(a)=\mathbf{d}_{a}((M,f,\Omega),\Phi)$, $\psi(b)=\mathbf{K}_{b}((M,f,\Omega),\Phi)$. Similarly, for any $c\in[0,w\cdot\rho_{\Omega}(M)-k]$, $b\in[0,\rho_{\Omega}(M)]$, it holds true that $\tau(c)=\mathbf{d}_{c}((N,h,\Delta),\Theta)$, $\eta(b)=\mathbf{K}_{b}((N,h,\Delta),\Theta)$. Therefore, an application of Theorem 3.2 immediately leads to the desired result.
\end{proof}

\setlength{\parindent}{2em}
We end this section with the following remark, where we use Theorem 5.1 to recover the Wei-type duality theorems for $(q,m)$-demi-polymatroids defined over vector spaces.

\setlength{\parindent}{0em}
\begin{remark}
Let $\mathds{F}$ be a field, $E$ a nonempty finite set. Set $R=S=U=\mathds{F}$, $M=N=\mathds{F}^{E}$. Furthermore, let $\varpi=\langle~,~\rangle$ denote the standard inner product on $\mathds{F}^{E}$, i.e.,
$$\mbox{$\varpi(\alpha,\beta)=\langle\alpha,\beta\rangle=\sum_{e\in E}\alpha_{(e)}\cdot\beta_{(e)}$ for all $\alpha,\beta\in\mathds{F}^{E}$},$$
and let $\Phi=\mathcal{P}(M)$. Then, Theorem 5.1 coincides with the Wei-type duality theorems for $(q,m)$-demi-polymatroids established in [7, Theorem 4] and [18, Theorem 17].
\end{remark}

\section{Generalized weights and profiles of codes with various metrics}

\setlength{\parindent}{2em}
In this section, we will consider codes defined over modules \cite{39,41} and treat their generalized weights and profiles with respect to poset metric, Gabidulin-Roth rank metric and Delsarte rank metric under the same ambient space. It turns out that given a code, each of the aforementioned metrics gives rise to an associated $w$-demi-matroid or $w$-demi-polymatroid, and the generalized weights/profiles of the code can be redefined as those of the associated $w$-demi-matroid or $w$-demi-polymatroid. The aforementioned approach was first introduced in \cite{5,6} for linear codes over division rings, and was then extended to Gabidulin-Roth and Delsate rank metric codes in \cite{7,18,21,26,38}. Our presentation will be in terms of left modules, which can be readily translated to one in terms of right modules. Throughout this section, we fix the following notations:

\setlength{\parindent}{0em}
\begin{itemize}
\item $R$ is a ring, and $M$ is a left $R$-module with a composition series.
\item $E$ is nonempty finite set with $|E|=m$.
\item $\mathcal{P}(M)$ and $\mathcal{P}(M^{E})$ are the sets of all the left $R$-submodules of $M$ and of $M^{E}$, respectively. Any left $R$-submodule of $M^{E}$ will be referred to as a \textit{code}.
\item For any $J\subseteq E$, the code $\delta(J)$ is defined as
\begin{equation}\delta(J)=\{\alpha\mid \alpha\in M^{E}~s.t.~\forall~i\in E-J,~\alpha_{(i)}=0\}.\end{equation}
\end{itemize}

\setlength{\parindent}{2em}
Rather than a single code, We will treat a family of codes, and hence the following notations are also introduced:

\setlength{\parindent}{0em}
\begin{itemize}
\item $I$ is a nonempty set, and $(u{(l)}\mid l\in I)$ is a family of positive integers.
\item $\mathbf{C}=\left(\mathbf{C}(l,j)\mid l\in I,~j\in[1,u(l)]\right)$ is a family of codes such that
\begin{equation}\text{$\mathbf{C}(l,u(l))\subseteq \mathbf{C}(l,u(l)-1)\subseteq\cdots\subseteq \mathbf{C}(l,1)$ for all $l\in I$}.\end{equation}
\item $\Omega$ is a nonempty collection of simple left $R$-modules with $w\triangleq\rho_{\Omega}(M)\geqslant\nolinebreak1$.
\item $k\in\mathbb{N}$ is defined as
$$\mbox{$k=\max\left\{\sum_{i=1}^{u(l)}(-1)^{i-1}\rho_{\Omega}(\mathbf{C}(l,i))\mid l\in I\right\}$}.$$
\end{itemize}

\setlength{\parindent}{2em}
We note that in terms of \textit{code flags} (\cite{5,18}), (6.2) says that for any $l\in I$, $(\mathbf{C}(l,1),\dots,\mathbf{C}(l,u(l)))$ is a code flag. By Lemma 5.1, we also have $$\mbox{$\rho_{\Omega}(M^{E})=wm$ and $k\in[0,wm]$.}$$

\subsection{Demi-matroids and demi-polymatroids arising from $\mathbf{C}$}

First of all, we define $f_0:\mathcal{P}(M^{E})\longrightarrow \mathds{Z}$ as
\begin{equation}\mbox{$f_0(L)=\max\left\{\sum_{i=1}^{u(l)}(-1)^{i-1}\rho_{\Omega}(\mathbf{C}(l,i)\cap L)\mid l\in I\right\}$.}\end{equation}
Next, we define $f_1:2^{E}\longrightarrow \mathds{Z}$ as
\begin{equation}\mbox{$f_1(J)=f_0(\delta(J))=\max\left\{\sum_{i=1}^{u(l)}(-1)^{i-1}\rho_{\Omega}(\mathbf{C}(l,i)\cap \delta(J))\mid l\in I\right\}$.}\end{equation}
Finally, we define $f_2:\mathcal{P}(M)\longrightarrow \mathds{Z}$ as
\begin{equation}\mbox{$f_2(W)=f_0(W^{E})=\max\left\{\sum_{i=1}^{u(l)}(-1)^{i-1}\rho_{\Omega}(\mathbf{C}(l,i)\cap W^{E})\mid l\in I\right\}$.}\end{equation}

With the help of Lemma 5.1, we derive the following proposition via some straightforward computation.

\setlength{\parindent}{0em}
\begin{proposition}
{\bf{(1)}}\,\,$(M^{E},f_0,\Omega)$ is a $1$-demi-polymatroid with $f_0(M^{E})=\nolinebreak k$.

{\bf{(2)}}\,\,$(E,f_1)$ is a $w$-demi-matroid with $f_1(E)=k$.

{\bf{(3)}}\,\,$(M,f_2,\Omega)$ is an $m$-demi-polymatroid with $f_2(M)=k$.
\end{proposition}

\setlength{\parindent}{2em}
Each demi-matroid/demi-polymatroid in Proposition 6.1 leads to a corresponding generalized weight/profile of $\mathbf{C}$, as detailed in the following definition.

\setlength{\parindent}{0em}
\begin{definition}
{\bf{(1)}}\,\,Let $\Phi\subseteq\mathcal{P}(M^{E})$ such that $\{0\}\in\Phi$, $M^{E}\in\Phi$ and $((\Phi,\subseteq),(\rho_{\Omega})\mid_{\Phi})$ is an abundance. For any $a\in [0,k]$, the $a$-th Gabidulin-Roth generalized rank weight of $(\mathbf{C},\Omega,\Phi)$, denoted by ${\mathbf{d}^{\mathbf{G}}}_{a}(\mathbf{C},\Omega,\Phi)$, is defined as

$${\mathbf{d}^{\mathbf{G}}}_{a}(\mathbf{C},\Omega,\Phi)=\mathbf{d}_{a}((M^{E},f_0,\Omega),\Phi),$$
and for any $b\in [0,wm]$, the $b$-th Gabidulin-Roth profile of $(\mathbf{C},\Omega,\Phi)$, denoted by ${\mathbf{K}^{\mathbf{G}}}_{b}(\mathbf{C},\Omega,\Phi)$, is defined as
$${\mathbf{K}^{\mathbf{G}}}_{b}(\mathbf{C},\Omega,\Phi)=\mathbf{K}_{b}((M^{E},f_0,\Omega),\Phi).$$

{\bf{(2)}}\,\,$(M^{E},f_0,\Omega)$ is said to be \textit{the $1$-demi-polymatroid associated to $(\mathbf{C},\Omega)$ with respect to Gabidulin-Roth rank metric}.

{\bf{(3)}}\,\,Consider a poset $\mathbf{P}=(E,\preccurlyeq_{\mathbf{P}})$. For any $a\in [0,k]$, the $a$-th generalized weight of $(\mathbf{C},\Omega,\mathbf{P})$, denoted by $\mathbf{d}_{a}(\mathbf{C},\Omega,\mathbf{P})$, is defined as
$$\mathbf{d}_{a}(\mathbf{C},\Omega,\mathbf{P})=\mathbf{d}_{a}(f_1,\mathbf{P}),$$
and for any $b\in [0,m]$, the $b$-th profile of $(\mathbf{C},\Omega,\mathbf{P})$, denoted by $\mathbf{K}_{b}(\mathbf{C},\Omega,\mathbf{P})$, is defined as
$$\mathbf{K}_{b}(\mathbf{C},\Omega,\mathbf{P})=\mathbf{K}_{b}(f_1,\mathbf{P}).$$

{\bf{(4)}}\,\,$(E,f_1)$ is said to be \textit{the $w$-demi-matroid associated to $(\mathbf{C},\Omega)$ with respect to poset metric}.

{\bf{(5)}}\,\,Let $\Phi\subseteq\mathcal{P}(M)$ such that $\{0\}\in\Phi$, $M\in\Phi$ and $((\Phi,\subseteq),(\rho_{\Omega})\mid_{\Phi})$ is an abundance. For any $a\in [0,k]$, the $a$-th Delsarte generalized rank weight of $(\mathbf{C},\Omega,\Phi)$, denoted by ${\mathbf{d}^{\mathbf{R}}}_{a}(\mathbf{C},\Omega,\Phi)$, is defined as
$${\mathbf{d}^{\mathbf{R}}}_{a}(\mathbf{C},\Omega,\Phi)=\mathbf{d}_{a}((M,f_2,\Omega),\Phi),$$
and for any $b\in [0,w]$, the $b$-th Delsarte profile of $(\mathbf{C},\Omega,\Phi)$, denoted by ${\mathbf{K}^{\mathbf{R}}}_{b}(\mathbf{C},\Omega,\Phi)$, is defined as
$${\mathbf{K}^{\mathbf{R}}}_{b}(\mathbf{C},\Omega,\Phi)=\mathbf{K}_{b}((M,f_2,\Omega),\Phi).$$

{\bf{(6)}}\,\,$(M,f_2,\Omega)$ is said to be \textit{the $m$-demi-polymatroid associated to $(\mathbf{C},\Omega)$ with respect to Delsarte rank metric}.
\end{definition}

\setlength{\parindent}{2em}
Adopting the notations in Definition 6.1, in the following remark, we show how Definition 6.1 recovers some known generalized weights/profiles of codes endowed with Gabidulin-Roth rank metric, poset metric and Delsarte rank metric.

\setlength{\parindent}{0em}
\begin{remark}
{\bf{(1)}}\,\,Following (1) of Definition 6.1, we let $\mathds{F}$ be a finite field. Suppose that $R$ is a finite field extension of $\mathds{F}$ and $M=R$. Then, any code $V\subseteq R^{E}$ is a de facto Gabidulin-Roth rank metric code (see \cite{7,15}). We now let $\Phi\subseteq\mathcal{P}(R^{E})$ denote the following collection of codes:
\begin{equation}\Phi=\{V\mid \text{$V$ is spanned by $V\cap\mathds{F}^{E}$ as an $R$-vector space}\}.\end{equation}
Assume that $\mathbf{C}$ is a code flag consists of two codes $C_1$ and $C_2$ with $C_2\subseteq C_1$. Then, it is straightforward to verify that for any $a\in [0,k]$, ${\mathbf{d}^{\mathbf{G}}}_{a}(\mathbf{C},\Omega,\Phi)$ is equal to
$$\min\left\{\dim_{R}(L)\mid L\in \Phi,~\dim_{R}(C_1\cap L)-\dim_{R}(C_2\cap L)=a\right\},$$
and for any $b\in [0,m]$, ${\mathbf{K}^{\mathbf{G}}}_{b}(\mathbf{C},\Omega,\Phi)$ is equal to
$$\max\{\dim_{R}(C_1\cap L)-\dim_{R}(C_2\cap L)\mid L\in\Phi,~\dim_{R}(L)=b\},$$
and so ${\mathbf{d}^{\mathbf{G}}}_{a}(\mathbf{C},\Omega,\Phi)$ and ${\mathbf{K}^{\mathbf{G}}}_{b}(\mathbf{C},\Omega,\Phi)$ coincide with the notions of relative generalized rank weights and relative dimension/intersection profiles of Gabidulin-Roth rank metric codes, respectively (see [27, Definitions 4 and 5] and \cite{15}).

{\bf{(2)}}\,\,Consider a poset $\mathbf{P}=(E,\preccurlyeq_{\mathbf{P}})$. First, we assume that $\mathbf{C}$ consists of one code, say, $C\subseteq M^{E}$. Then, it is straightforward to verify that for any $r\in[0,k]$, we have
$$\mathbf{d}_{r}(C,\Omega,\mathbf{P})=\min\{|\langle \chi(D)\rangle_{\mathbf{P}}|\mid D\in\mathcal{P}(M^{E}),~D\subseteq C,~\rho_{\Omega}(D)=r\},$$
where for any $D\subseteq M^{E}$, $\chi(D)\subseteq E$ is defined as in (1.2) (with $\mathbb{F}$ replaced by $M$). Therefore, for the case that $R$ is a field and $M=R$, Definition 6.1 coincides with the definition of generalized weights for a code with a poset metric (\cite{9}, \cite{32}). Furthermore, with $\mathbf{P}$ set to be an anti-chain, Definition 6.1 coincides with the GHWs of $C$ as in (1.1).

\hspace*{6mm}Now consider the case that $\mathbf{C}$ is a code flag (i.e., $|I|=1$), $R$ is a division ring, $M=R$ and $\mathbf{P}$ is an anti-chain. Then, Definition 6.1 coincides with the generalized weights and profiles for a code flag ([5, Section III]); if we further assume that $\mathbf{C}$ is a code flag consists of two codes $C_1$ and $C_2$ with $C_2\subseteq C_1$, Definition 6.1 then boils down to the relative generalized Hamming weight and relative dimension/length profile defined in \cite{28}.

{\bf{(3)}}\,\,Assume that $E=[1,m]$, $R$ is a field and $M=R^{w}$. With each vector of $M$ in its column form, we identify $M^{E}$ with the set of all matrices over $R$ with $w$ rows and $m$ columns. Then, a code $C\subseteq M^{E}$ becomes a Delsarte rank metric code (see [20, Definition 1.1]). Moreover, we set $\Phi$ in Definition 6.2 as $\mathcal{P}(M)$.

\hspace*{6mm}First, we let $w=m$, $I=\{1,2\}$, $u(1)=u(2)=1$, and for a fixed code $C$, define $\mathbf{C}(1,1)=C$, $\mathbf{C}(2,1)=\{\theta^{T}\mid \theta\in C\}$. Then, Definition 6.1 coincides with Ravagnani's definition for generalized rank weights of $C$, which is originally proposed by an optimal anticodes approach (see [36, Definition 23], [20, Remark 5.8]).

\hspace*{6mm}Second, we consider the case that $\mathbf{C}$ is a code flag. Then, Definition 6.1 becomes generalized rank weights of Delsarte rank metric code flags [18, Definition 46]. If the code flag consists of two codes, then Definition 6.1 becomes the relative generalized matrix weights and relative dimension/rank support profiles proposed in [31, Definitions 10 and 11]. If $\mathbf{C}$ consists of one code, say, $C\subseteq M^{E}$, then Definition 6.1 becomes the generalized matrix weight of $C$ [31, Definitions 10], and moreover, as long as $k\geqslant1$, ${\mathbf{d}^{\mathbf{R}}}_{1}(C,\Omega,\mathcal{P}(M))=\min\{\rank(\theta)\mid \theta\in C,~\theta\neq0\}$ is exactly the minimal rank distance of $C$ (see [20, Definition 3.1]), where $\rank$ denotes the rank of a matrix.
\end{remark}

\section{Wei-type duality theorems for codes}

\setlength{\parindent}{2em}
In this section, we prove Wei-type duality theorems for codes with Gabidulin-Roth metric, poset metric, Delsarte rank metric and Generalized Hamming weight with respect to rank (see \cite{23}). Throughout this section, we fix the following notations:

\setlength{\parindent}{0em}
\begin{itemize}
\item $R$ and $S$ are rings, and $U$ is an $R$-$S$ bimodule satisfying (5.5) and \nolinebreak(5.6).
\item $M$ is a left $R$-module with a composition series, and $N$ is a right $S$-module.
\item $E$ is a nonempty finite set with $|E|=m$. Furthermore, any left $R$-submodule of $M^{E}$ or right $S$-submodule of $N^{E}$ will be referred to as a \textit{left linear code} or a \textit{right linear code}, respectively.
\item $\varpi:M\times N\longrightarrow U$ is a non-degenerated bilinear map.
\item $\langle~,~\rangle:M^{E}\times N^{E}\longrightarrow U$ is the non-degenerated bilinear map defined \nolinebreak as
$$\mbox{$\langle\alpha,\beta\rangle=\sum_{e\in E}\varpi(\alpha_{(e)},\beta_{(e)})$}.$$
\item $\Omega$ is a nonempty collection of simple left $R$-modules with $w\triangleq\rho_{\Omega}(M)\geqslant\nolinebreak1$.
\item $\Delta$ is the collection of simple right $S$-modules defined as
$$\Delta=\{{\Hom_{R}}(X,U)\mid X\in \Omega\}.$$
\item $I$ is a nonempty set, and $(u(l)\mid l\in I)$ is a family of {\textbf{odd}} positive integers.
\item $\mathbf{C}=\left(\mathbf{C}(l,j)\mid l\in I,~j\in[1,u(l)]\right)$ such that $(\mathbf{C}(l,1),\dots,\mathbf{C}(l,u(l)))$ is a left linear code flag for all $l\in I$.
\item $\mathbf{D}=\left(\mathbf{D}(l,j)\mid l\in I,~j\in[1,u(l)]\right)$ is defined as
$$\mbox{$\mathbf{D}(l,j)=\mathbf{C}(l,u(l)+1-j)^{\bot}$ for all $l\in I$, $j\in[1,u(l)]$}.$$
\end{itemize}

\setlength{\parindent}{2em}
We further assume that there uniquely exists $k\in[0,wm]$ such that
\begin{equation}\mbox{$\sum_{i=1}^{u(l)}(-1)^{i-1}\rho_{\Omega}(\mathbf{C}(l,i))=k$ for all $l\in I$}.\end{equation}

\setlength{\parindent}{2em}
Naturally, $\mathbf{D}$ can be regarded as the dual of $\mathbf{C}$. Furthermore, with the help of Lemma 5.1, (5.7) and (5.8), we derive the following lemma via some straightforward computation.

\setlength{\parindent}{0em}
\begin{lemma}
For any $l\in I$, $(\mathbf{D}(l,1),\dots,\mathbf{D}(l,u(l)))$ is a right linear code flag with
$$\mbox{$\sum_{i=1}^{u(l)}(-1)^{i-1}\lambda_{\Delta}(\mathbf{D}(l,i))=wm-k$}.$$
Furthermore, for any $l\in I$ and $V\in\mathcal{P}(N^{E})$, it holds that
$$\hspace*{-6mm}\mbox{$\left(\sum_{i=1}^{u(l)}(-1)^{i-1}\lambda_{\Delta}(\mathbf{D}(l,i)\cap V)\right)-\left(\sum_{i=1}^{u(l)}(-1)^{i-1}\rho_{\Omega}(\mathbf{C}(l,i)\cap (^{\bot}V))\right)=\lambda_{\Delta}(V)-k$}.$$
\end{lemma}

\setlength{\parindent}{2em}
\begin{remark}
The assumption ``$u(l)$ is odd for all $l\in I$'' is essential and cannot be removed, since otherwise Lemma 7.1 (which we heavily rely on) would fail to hold true (see [5, Theorem 10], [18, Proposition 32]).

We note that Ghorpade and Johnsen proposed an alternative approach to establish modified Wei-type duality theorems for Delsarte rank metric code flags of even length in [18, Theorem 33] (in particular, for relative generalized matrix weights; see \cite{31}). Their approach, which is based on Wei-type duality theorems for code flags of odd length, can also be adapted to poset metric and Gabidulin-Roth rank metric (in particular, to relative generalized Hamming weights; see \cite{28}).
\end{remark}

\subsection{Wei-type duality theorems for Gabidulin-Roth rank metric, poset metric and Delsarte rank metric}

\setlength{\parindent}{2em}
Throughout this subsection, we let $(M^{E},f_0,\Omega)$ and $(N^{E},h_0,\Delta)$ denote the associated $1$-demi-polymatroids of $(\mathbf{C},\Omega)$ and of $(\mathbf{D},\Delta)$, respectively, both with respect to Gabidulin-Roth rank metric; and let $(E,f_1)$ and $(E,h_1)$ denote the associated $w$-demi-matroids of $(\mathbf{C},\Omega)$ and of $(\mathbf{D},\Delta)$, respectively, both with respect to poset metric; and let $(M,f_2,\Omega)$ and $(N,h_2,\Delta)$ denote the associated $m$-demi-polymatroids of $(\mathbf{C},\Omega)$ and of $(\mathbf{D},\Delta)$, respectively, both with respect to Delsarte rank metric.

\setlength{\parindent}{0em}
\begin{proposition}
{\bf{(1)}}\,\,For any $V\in\mathcal{P}(N^{E})$, $h_0(V)=f_0(^{\bot}V)+\lambda_{\Delta}(V)-k$.

{\bf{(2)}}\,\,For any $J\subseteq E$, $h_1(J)=f_1(E-J)+w|J|-k$.

{\bf{(3)}}\,\,For any $L\in\mathcal{P}(N)$, $h_2(L)=f_2(^{\bot}L)+m\cdot\lambda_{\Delta}(L)-k$.
\end{proposition}

\begin{proof}
(1)\,\,Let $V\in\mathcal{P}(N^{E})$. Then, by (6.3), we have
$$\mbox{$f_0({^{\bot}V})=\max\left\{\sum_{i=1}^{u(l)}(-1)^{i-1}\rho_{\Omega}(\mathbf{C}(l,i)\cap(^{\bot}V))\mid l\in I\right\}$},$$
$$\mbox{$h_0(V)=\max\left\{\sum_{i=1}^{u(l)}(-1)^{i-1}\lambda_{\Delta}(\mathbf{D}(l,i)\cap V)\mid l\in I\right\}$},$$
which, together with Lemma 7.1, immediately implies (1).

(2)\,\,For any $B\subseteq E$, define $\delta(B)\subseteq M^{E}$ as in (6.1), and let $\varepsilon(B)\subseteq N^{E}$ denote the set of all the codewords of $N^{E}$ whose positions outside of $B$ are zeros. For any $J\subseteq E$, by (1) and (6.4), we have
$$h_1(J)=h_0(\varepsilon(J))=f_0(^{\bot}\varepsilon(J))+\lambda_{\Delta}(\varepsilon(J))-k.$$
Noticing that $^{\bot}\varepsilon(J)=\delta(E-J)$, $\lambda_{\Delta}(N)=\rho_{\Omega}(M)=w$ (by $^{\bot}N=\{0\}$ and (5.8)) and $\lambda_{\Delta}(\varepsilon(J))=\lambda_{\Delta}(N)\cdot|J|=w|J|$ (Lemma 5.1), we conclude that
$$h_1(J)=f_0(\delta(E-J))+w|J|-k=f_1(E-J)+w|J|-k,$$
proving (2).

(3)\,\,For any $L\in\mathcal{P}(N)$, by (1) and (6.5), we have
$$h_2(L)=h_0(L^{E})=f_0(^{\bot}(L^{E}))+\lambda_{\Delta}(L^{E})-k.$$
Since $^{\bot}(L^{E})=(^{\bot}L)^{E}$, we have $f_0(^{\bot}(L^{E}))=f_0((^{\bot}L)^{E})=f_2(^{\bot}L)$, which, together with $\lambda_{\Delta}(L^{E})=m\cdot\lambda_{\Delta}(L)$ (Lemma 5.1), immediately implies (3).
\end{proof}

\setlength{\parindent}{2em}
Combining Theorem 5.1, (1) of Definition 6.1 and (1) of Proposition 7.1, we have the following Wei-type duality theorem for Gabidulin-Roth rank metric.

\setlength{\parindent}{0em}
\begin{theorem}
{Fix $\Phi\subseteq\mathcal{P}(M^{E})$ such that $\{0\},M^{E}\in\Phi$, $((\Phi,\subseteq),(\rho_{\Omega})\mid_{\Phi})$ is an abundance, and let $\Theta=\{X^{\bot}\mid X\in\Phi\}$. Then, we have

{\bf{(1)}}\,\,For any $b\in[0,wm]$, ${\mathbf{K}^{\mathbf{G}}}_{b}(\mathbf{D},\Delta,\Theta)={\mathbf{K}^{\mathbf{G}}}_{wm-b}(\mathbf{C},\Omega,\Phi)+b-k$;

{\bf{(2)}}\,\,$\{{\mathbf{d}^{\mathbf{G}}}_{a}(\mathbf{C},\Omega,\Phi)\mid a\in[1,k]\}$ and $\{wm+1-{\mathbf{d}^{\mathbf{G}}}_{c}(\mathbf{D},\Delta,\Theta)\mid c\in[1,wm-k]\}$ form a partition of $[1,wm]$.
}
\end{theorem}

\setlength{\parindent}{2em}
Combining Theorem 4.2, (3) of Definition 6.1 and (2) of Proposition 7.1, we have the following Wei-type duality theorem for poset metric.

\setlength{\parindent}{0em}
\begin{theorem}
{Let $\mathbf{P}=(E,\preccurlyeq_{\mathbf{P}})$ be a poset. Then, we have

{\bf{(1)}}\,\,For any $b\in [0,m]$, $\mathbf{K}_{b}(\mathbf{D},\Delta,\mathbf{\overline{P}})=\mathbf{K}_{m-b}(\mathbf{C},\Omega,\mathbf{P})+wb-k$;

{\bf{(2)}}\,\,For any $\gamma\in\mathds{Z}$, let $\mathcal{A}_{(\gamma)}=\{\mathbf{d}_{a}(\mathbf{C},\Omega,\mathbf{P})\mid a\in[1,k],a\equiv\gamma+k~(\bmod~w)\}$, $\mathcal{B}_{(\gamma)}=\{m+1-\mathbf{d}_{c}(\mathbf{D},\Delta,\mathbf{\overline{P}})\mid c\in[1,wm-k],~c\equiv\gamma~(\bmod~w)\}$. Then, we have $\mathcal{A}_{(\gamma)}\cap \mathcal{B}_{(\gamma)}=\emptyset$, $\mathcal{A}_{(\gamma)}\cup \mathcal{B}_{(\gamma)}=[1,m]$.
}
\end{theorem}

\setlength{\parindent}{2em}
Combining Theorem 5.1, (5) of Definition 6.1 and (3) of Proposition 7.1, we have the following Wei-type duality theorem for Delsarte rank metric.

\setlength{\parindent}{0em}
\begin{theorem}
Let $\Phi\subseteq\mathcal{P}(M)$ such that $\{0\},M\in\Phi$, $((\Phi,\subseteq),(\rho_{\Omega})\mid_{\Phi})$ is an abundance, and let $\Theta=\{X^{\bot}\mid X\in\Phi\}$. Then, we have

{\bf{(1)}}\,\,For any $b\in[0,w]$, ${\mathbf{K}^{\mathbf{R}}}_{b}(\mathbf{D},\Delta,\Theta)={\mathbf{K}^{\mathbf{R}}}_{w-b}(\mathbf{C},\Omega,\Phi)+mb-k$;

{\bf{(2)}}\,\,For any $\gamma\in\mathds{Z}$, let $\mathcal{A}_{(\gamma)}=\{{\mathbf{d}^{\mathbf{R}}}_{a}(\mathbf{C},\Omega,\Phi)\mid a\in[1,k],a\equiv\gamma+k~(\bmod~m)\}$, $\mathcal{B}_{(\gamma)}=\{w+1-{\mathbf{d}^{\mathbf{R}}}_{c}(\mathbf{D},\Delta,\Theta)\mid c\in[1,wm-k],~c\equiv\gamma~(\bmod~m)\}$. Then, we have $\mathcal{A}_{(\gamma)}\cap \mathcal{B}_{(\gamma)}=\emptyset$, $\mathcal{A}_{(\gamma)}\cup \mathcal{B}_{(\gamma)}=[1,w]$.
\end{theorem}

\setlength{\parindent}{2em}
In the following remark, we show how Theorems 7.1-7.3 recover some known Wei-type duality theorems for codes.

\setlength{\parindent}{0em}
\begin{remark}
{\bf{(1)}}\,\,Let $\mathds{F}$ be a finite field. Suppose that $R$ is a finite field extension of $\mathds{F}$, and set $U=M=N=R=S$. Let $\varpi$ denote the multiplication within $R$. Then, $\langle~,~\rangle$ becomes the standard inner product of $R^{E}$. Now define $\Phi\subseteq\mathcal{P}(R^{E})$ as in (6.6). By [15, Lemma III.1], we have $\{V^{\bot}\mid V\in\Phi\}=\Phi$. Hence, Theorem 7.1 recover [15, Theorem I.3] for Gabidulin-Roth rank metric codes.

{\bf{(2)}}\,\,Let $R$ be a quasi-Frobenius ring (see \cite{1,11,39,41}), $U=M=N=S=R$, and $\varpi$ be multiplication within $R$. Then, $\langle~,~\rangle$ becomes the standard inner product of $R^{E}$. We note that by [11, Theorem 58.6], $U$ satisfies (5.5) and (5.6). Then, Theorem 7.2 recovers [3, Lemma 2.2], [6, Theorem 11] and [32, Theorem 2] which are stated for codes over fields or division rings. Since a Galois ring is quasi-Frobenius, Theorem 7.2 can also be regarded as a poset metric generalization of [2, Theorem 2] for Galois ring linear codes.

{\bf{(3)}}\,\,As in (3) of Remark 6.1, we let $E=[1,m]$, $R$ be a field, $U=R=S$, $M=N=R^{w}$, and $\varpi$ be the standard inner product of $R^{w}$. Then, for any two matrices $\gamma,\theta\in M^{E}$, we have $\langle\gamma,\theta\rangle=\tr(\gamma\cdot\theta^{T})$ ([20, Definition 4.1]). Set $\Phi=\mathcal{P}(M)$. Then, for the case that $\mathbf{C}$ consists of one code, Theorem 7.3 becomes [31, Proposition 65]. If $|I|=1$, then Theorem 7.3 becomes the Wei-type duality theorem for a code flag, which alternatively follows from [18, Theorem 39]. If $m=w$, $I=\{1,2\}$, $u(1)=u(2)=1$, and $\mathbf{C}(2,1)=\{\theta^{T}\mid \theta\in \mathbf{C}(1,1)\}$, then Theorem 7.3 recovers [36, Corollary 38].
\end{remark}

\subsection{Wei-type duality theorem for generalized weights with respect to rank of modules over finite chain rings}

\setlength{\parindent}{2em}
In \cite{23}, Horimoto and Shiromoto prove a Wei-type duality theorem for Generalized Hamming weights with respect to rank (GHWR) (see [23, Theorem 3.12]). In this subsection, we study GHWR via a demi-matroid approach, and give a poset metric generalization of their result. This subsection is different from Section 7.1 in that the corresponding demi-matroid will be defined by rank of modules, instead of $\rho_{\Omega}$ or $\lambda_{\Delta}$.

Assume that $A$ is a finite chain ring, i.e., $A$ is a finite local ring, and the Jacobson radical of $A$ is a principle left ideal. For any finitely generated left (right) $A$-module $X$, we let $\rank(X)$ denote the minimum number of generators of $X$.

Throughout this subsection, we set $R=S=U=M=N=A$ and let $\varpi$ be multiplication within $A$. Then, $\langle~,~\rangle$ becomes the standard inner product of $A^{E}$. For any $A$-$A$-submodule $L\subseteq A^{E}$, thanks to [23, Proposition 2.1], the minimum number of generators of $L$ as a left $A$-module is equal to the minimum number of generators of $L$ as a right $A$-module. For any $J\subseteq E$, we define $\delta(J)\subseteq A^{E}$ as in (6.1). Then, $\delta(J)$ is an $A$-$A$-submodule with $\rank(\delta(J))=|J|$.

\setlength{\parindent}{2em}
With respect to $\rank$, every linear code gives rise to a demi-matroid, as detailed in the following proposition.

\setlength{\parindent}{0em}
\begin{proposition}
Let $C\subseteq A^{E}$ be a left (or right) linear code with $t=\rank(C)$, and define $f:2^{E}\longrightarrow \mathds{Z}$ as $f(J)=\rank(C\cap\delta(J))$. Then, $(E,f)$ is a demi-matroid with $f(E)=t$.
\end{proposition}

\begin{proof}
Since $A$ is a finite chain ring, for any finitely generated left (right) $A$-module $W$ and $A$-submodule $V\subseteq W$, we have
\begin{equation}0\leqslant\rank(W)-\rank(V)\leqslant \rank(W/V).\end{equation}
Now without loss of generality, we assume $C$ is a right linear code. Since $\delta(\emptyset)=\{0\}$, $\delta(E)=A^{E}$, we have $f(\emptyset)=0$, $f(E)=\rank(C)=t$. For any $B,D$ with $B\subseteq D\subseteq E$, we have $\delta(B)\subseteq\delta(D)$, which, together with (7.2), yields that
$$\hspace*{-4mm}0\leqslant f(D)-f(B)\leqslant \rank((C\cap\delta(D))/(C\cap\delta(B)))\leqslant\rank(\delta(D)/\delta(B))=|D|-|B|.$$
Therefore $(E,f)$ is a demi-matroid, completing the proof.
\end{proof}

\setlength{\parindent}{2em}
Based on Proposition 7.4, we give the following definition.

\setlength{\parindent}{0em}

\begin{definition}
Let $C\subseteq A^{E}$ be a left (or right) linear code with $t=\rank(C)$, and define $f:2^{E}\longrightarrow \mathds{Z}$ as $f(J)=\rank(C\cap\delta(J))$. Then, $(E,f)$ is called the associated demi-matroid of $C$ with respect to rank. Moreover, given a poset $\mathbf{P}=(E,\preccurlyeq_{\mathbf{P}})$, for any $a\in [0,t]$, the $a$-th generalized weight of $(C,\mathbf{P})$ with respect to rank, denoted by $\mathbf{\widetilde{d}}_{a}(C,\mathbf{P})$, is defined as
$$\mathbf{\widetilde{d}}_{a}(C,\mathbf{P})=\mathbf{d}_{a}(f,\mathbf{P})=\min\{|J|\mid J\in\mathcal{I}(\mathbf{P}),~a\leqslant \rank(C\cap\delta(J))\},$$
and for any $b\in [0,m]$, the $b$-th profile of $(C,\mathbf{P})$ with respect to rank, denoted by $\mathbf{\widetilde{K}}_{b}(C,\mathbf{P})$, is defined as
$$\mathbf{\widetilde{K}}_{b}(C,\mathbf{P})=\mathbf{K}_{b}(f,\mathbf{P})=\max\{\rank(C\cap\delta(J))\mid J\in\mathcal{I}(\mathbf{P}),~|J|=b\}.$$
\end{definition}

\setlength{\parindent}{2em}
The following Lemma shows that Definition 7.1 naturally extends the notion of GHWR (see [23, Definition 3.1]).

\begin{lemma}
Let $C\subseteq A^{E}$ be a left (or right) linear code with $t=\rank(C)$, and fix a poset $\mathbf{P}=(E,\preccurlyeq_{\mathbf{P}})$. Then, for any $a\in[0,t]$, we have
$$\mathbf{\widetilde{d}}_{a}(C,\mathbf{P})=\min\{|\langle \chi(D)\rangle_{\mathbf{P}}|\mid D~\mbox{ is a left (right) subcode of }~C,\rank(D)=a\},$$
where $\chi(D)$ is defined as in (1.2) (with $\mathbb{F}$ replaced by $A$). In particular, if $\mathbf{P}$ is set to be the anti-chain, then $\mathbf{\widetilde{d}}_{a}(C,\mathbf{P})$ is equal to the $a$-th GHWR of $C$.
\end{lemma}

\setlength{\parindent}{2em}
Now we consider the corresponding Wei-type duality theorem. The free code, i.e., the free $A$-submodule of $A^{E}$, will play a particularly important role in the discussion.

\setlength{\parindent}{2em}
Throughout the rest of this subsection, we let $C\subseteq A^{E}$ be a right linear code with $t=\rank(C)$, and let $M\subseteq A^{E}$ be a free right linear code such that $C\subseteq M$ and $t=\rank(M)$. We note that such $M$ always exists, and moreover, $^{\bot}M$ is a free left linear code with $\rank(^{\bot}M)=m-t$.

Let $(E,f)$ and $(E,h)$ be the associated demi-matroids of $M$ and $^{\bot}M$ with respect to rank, respectively. It turns out that $(E,f)$ and $(E,h)$ are dual demi-matroids in the sense of Proposition 4.1, which, albeit without mentioning demi-matroids, has already been observed in the proof of [23, Theorem 3.12].

\setlength{\parindent}{0em}
\begin{proposition}
For any $J\subseteq E$, $h(J)=f(E-J)+|J|-t$.
\end{proposition}

\begin{proof}
Let $J\subseteq E$. Since $M$ is free, by [23, Lemma 3.11], we have
$$\rank(M\cap\delta(E-J))=\rank(M)+\rank((^{\bot}M)\cap\delta(J))-\rank(\delta(J)).$$
Note that $\rank(\delta(J))=|J|$, $t=\rank(M)$, $\rank(M\cap\delta(E-J))=f(E-J)$ and $\rank((^{\bot}M)\cap\delta(J))=h(J)$, the proposition immediately follows.
\end{proof}

\setlength{\parindent}{2em}
We are now ready to state and prove the Wei-type duality theorem.

\setlength{\parindent}{0em}
\begin{theorem}
For any given poset $\mathbf{P}=(E,\preccurlyeq_{\mathbf{P}})$, it holds true that:

{\bf{(1)}}\,\,$\mathbf{\widetilde{d}}_{a}(C,\mathbf{P})=\mathbf{\widetilde{d}}_{a}(M,\mathbf{P})$ for any $a\in[0,t]$;

{\bf{(2)}}\,\,$\mathbf{\widetilde{K}}_{b}(C,\mathbf{P})=\mathbf{\widetilde{K}}_{b}(M,\mathbf{P})$ for any $b\in[0,m]$;

{\bf{(3)}}\,\,$\mathbf{\widetilde{K}}_{l}(^{\bot}M,\mathbf{\overline{P}})=\mathbf{\widetilde{K}}_{m-l}(C,\mathbf{P})+l-t$ for any $l\in[0,m]$;

{\bf{(4)}}\,\,$\{\mathbf{\widetilde{d}}_{u}(C,\mathbf{P})\mid u\in[1,t]\}$ and $\{m+1-\mathbf{\widetilde{d}}_{v}(^{\bot}M,\mathbf{\overline{P}})\mid v\in[1,m-t]\}$ form a partition of $[1,m]$.
\end{theorem}

\begin{proof}
(1)\,\,This follows from [23, Lemma 3.3 and Theorem 3.4]. We note that although these results were originally stated for GHWR, with the help of Lemma 7.2, one only needs to slightly modify the proofs in \cite{23} for poset metric.

(2)\,\,This follows from (1) and Lemma 2.1, since generalized weights and profiles with respect to rank for any linear code form a Galois connection.

(3) and (4)\,\,These two statements follow from Proposition 7.3 and Theorem 4.2, together with the proven parts (1) and (2).
\end{proof}

\setlength{\parindent}{2em}
\begin{remark}
Theorem 7.4 can be alternatively proved by using Proposition 7.3 and the Wei-type duality theorem for demi-matroids [6, Theorem 6], and can be regarded as a poset metric generalization of [23, Theorem 3.12]. We also note that Theorems 7.2 and 7.4 enable us to answer a question raised in [6, Section 3.4], namely, whether it is possible to find a poset metric generalization of Wei-type duality theorems established for codes over Galois rings and chain rings.
\end{remark}

\end{document}